\documentclass{article}

\usepackage{microtype}
\usepackage{graphicx}
\usepackage{subfigure}
\usepackage{booktabs} % for professional tables

\usepackage{hyperref}

\usepackage[accepted]{icml2023}

\usepackage{amsmath}
\usepackage{amssymb}
\usepackage{mathtools}
\usepackage{amsthm}

\usepackage[capitalize,noabbrev]{cleveref}

\theoremstyle{plain}
\newtheorem{theorem}{Theorem}[section]
\newtheorem{proposition}[theorem]{Proposition}
\newtheorem{lemma}[theorem]{Lemma}
\newtheorem{corollary}[theorem]{Corollary}
\theoremstyle{definition}
\newtheorem{definition}[theorem]{Definition}
\newtheorem{assumption}[theorem]{Assumption}
\theoremstyle{remark}

\usepackage[textsize=tiny]{todonotes}

\icmltitlerunning{Regret-Minimizing Double Oracle for Extensive-Form Games}

\begin{document}

\twocolumn[
\icmltitle{Regret-Minimizing Double Oracle for Extensive-Form Games}



\icmlsetsymbol{equal}{*}

\begin{icmlauthorlist}
\icmlauthor{Xiaohang Tang}{ucl}
\icmlauthor{Le Cong Dinh}{sot}
\icmlauthor{Stephen Marcus McAleer}{cmu}
\icmlauthor{Yaodong Yang}{pek}
\end{icmlauthorlist}

\icmlaffiliation{ucl}{Department of Statistical Science, University College London}
\icmlaffiliation{sot}{University of Southampton}
\icmlaffiliation{cmu}{Carnegie Mellon University}
\icmlaffiliation{pek}{Institute for AI, Peking University}

\icmlcorrespondingauthor{Yaodong Yang}{yaodong.yang@pku.edu.cn}

\icmlkeywords{Machine Learning, ICML}

\vskip 0.3in
]



\printAffiliationsAndNotice{} % otherwise use the standard text.

\begin{abstract}
By incorporating regret minimization, double oracle methods have demonstrated rapid convergence to Nash Equilibrium (NE) in normal-form games and extensive-form games, through algorithms such as online double oracle (ODO) and extensive-form double oracle (XDO), respectively. In this study, we further examine the theoretical convergence rate and sample complexity of such regret minimization-based double oracle methods, utilizing a unified framework called Regret-Minimizing Double Oracle. Based on this framework, we extend ODO to extensive-form games and determine its sample complexity. Moreover, we demonstrate that the sample complexity of XDO can be exponential in the number of information sets $|S|$, owing to the exponentially decaying stopping threshold of restricted games. To solve this problem, we propose the Periodic Double Oracle (PDO) method, which has the lowest sample complexity among regret minimization-based double oracle methods, being only polynomial in $|S|$. Empirical evaluations on multiple poker and board games show that PDO achieves significantly faster convergence than previous double oracle algorithms and reaches a competitive level with state-of-the-art regret minimization methods.
 
% Double oracle (DO) algorithms have been used successfully to solve large extensive-form games. In DO family, the Extensive-Form Double Oracle (XDO) is the best, showing faster convergence than state-of-the-art algorithms in games with small support Nash equilibrium (NE). However, in the games with large support NE, XDO failed since its sample complexity for finding approximate NE can grow exponentially with the number of information sets (infosets). In this paper, we present the first Double Oracle algorithm for EFGs whose sample complexity is only polynomially in the number of infosets, called Extensive-Form Online Double Oracle (XODO). XODO calculates the best response more frequently, which helps it to obtain a regret bound and thus make itself an anytime algorithm. In the empirical assessments, we show that a variant of XODO outperforms XDO in four games. This variant also outperforms the state-of-the-art algorithm in games with small support NE and is competitive in other games, suggesting that XODO not only retains the strengths of DO, but also remains robust in games without small support NE.
% Abstracts must be a single paragraph, ideally between 4--6 sentences long.
\end{abstract}

\section{Introduction}
Extensive-form games (EFGs) have been extensively employed in modeling sequential decision-making problems, including auctions, security games, and poker. However, solving such games remains challenging in real-world applications due to various complexities, including the game's large size and imperfect information. While linear programming (LP)~\citep{lp} and sequence-form LP~\citep{sequence_lp} can be used to compute the Nash equilibrium (NE) of EFGs~\citep{ne}, direct calculation of the exact NE via LP can become infeasible for large real-world games due to memory constraints and the high cost of matrix inversion. Therefore, more efficient methods are required to solve EFGs for real-world applications.

The Double Oracle (DO)~\citep{do} algorithm family has been developed to address the complexity of solving large Extensive-Form Games (EFGs) by solving a sequence of restricted games, where players can only select actions from a subset of pure strategies in the original game. These restricted games are typically much smaller than the original EFG, especially when the support of the Nash equilibrium (NE) is small ~\citep{do_small_sub1, do_small_sub2, sdo}. For instance, in symmetric normal-form games with random entries, the NE's support size is only half of the game size ~\citep{jonasson2004optimal}. By constructing a solution to the large game using only a small subset of pure strategies, the DO algorithm can effectively solve large EFGs. The DO algorithms can also be combined with deep reinforcement learning ~\citep{lanctot2017unified, muller2019generalized, mcaleer2022anytime, mcaleer2022self} to achieve state-of-the-art performance on large games like Stratego ~\citep{mcaleer2020pipeline} and Starcraft ~\citep{vinyals2019grandmaster}. 

Recent advances in game theory have led to the development of methods that integrate the Double Oracle (DO) algorithm with regret minimization techniques. For example, in large normal-form games, the Online Double Oracle (ODO) algorithm combines the DO approach with the Multiplicative Weights Update (MWU) ~\citep{mwu} regret minimizer for restricted games ~\citep{odo}. In extensive-form games, the Counterfactual Regret Minimization (CFR) algorithm has been successfully used to achieve superhuman performance in Texas Hold'em Poker and other games ~\citep{cfr, brown2019deep, mccfr,srm,mcaleer2022escher}. To further improve the performance of DO in this setting, the Extensive-Form DO (XDO) algorithm applies CFR to the DO framework and can support games with high-dimensional continuous action spaces ~\citep{xdo}. However, theoretical questions regarding the convergence speed, expected iterations, and sample complexity of DO methods in EFGs have yet to be thoroughly investigated and remain an open area of research ~\citep{sdo}.

In this study, a general theoretical framework is presented to investigate algorithms that combine DO and regret minimization, providing their expected number of iterations and sample complexity to reach $\epsilon$-NE. This is to our knowledge the first work analyzing the theoretical convergence rate of DO in EFGs. The paper presents two applications of the framework. In the first example, we extend ODO to solve EFGs and determine the corresponding sample complexity. In the second example, we prove that the stopping condition for restricted games in the Extensive-Form DO (XDO) algorithm may lead to a worst-case convergence in a number of iterations exponential in the number of information sets $|S|$. To address this issue, we propose Periodic Double Oracle (PDO), which only has a polynomial sample complexity in $|S|$, lower than that of both ODO and XDO. Empirical assessments on typical poker and board games demonstrate that PDO achieves much faster convergence compared to XDO and ODO, reaching a level competitive with state-of-the-art regret minimization methods.

\section{Preliminaries}
\subsection{Two-Player Zero-Sum Extensive-Form Games}
In this paper, we focus on Two-Player Zero-Sum Extensive-Form Games (\textbf{EFGs}) with perfect recall. Our notation is based on that of \citet{noamthesis}. EFGs are represented by a game tree where nodes correspond to players $i\in \mathcal{P}=\{1,2\}$. Imperfect Information EFGs employ \textbf{chance player} c to model the stochastic events like card dealing in Poker. \textbf{History} ($h$) is a sequence of actions the players took and events uniquely attached to a node on the game tree. $A(h)$ is the set of available actions at $h$, and $P(h)$ denotes the player who needs to make a decision at $h$. Terminal histories, which are all in the set $Z$, are attached to nodes where the game's payoff can be determined. The payoff at the terminal history $z\in Z$ is treated as value at $z$, denoted as $v_i(z)$. The range of the payoff is represented by $\Delta$.

\textbf{Information Set (Infoset/Infostate $s_i$)} of each player $i \in \mathcal{P}$ is the set of indistinguishable histories from player $i$'s perspective. The set $S_i$ contains infosets where player $i$ must make decisions, and $S$ represents the union of information sets for all players: $S=\cup_{i\in \mathcal{P}}S_i$. The set $A(s_i)$ includes all available actions of player $i$ at $s_i$. The \textbf{Strategy} of player $i$ is represented by $\pi_i$, with $\pi_i(s_i,a)$ representing the probability of player $i$'s taking action $a \in A(s_i)$ at the information set $s_i$. Joint strategy $\pi=(\pi_1, \pi_2)$. \textbf{Reaching probability }$x^{\pi}(h)$ is the probability to reach $h$ when players use joint strategy $\pi$. Specifically, $x^\pi(h) = \prod_{i\in \mathcal{P} \bigcup \{c\}}x^\pi_i(h)$, where $x^\pi_i(h)=\prod_{h^{'}\cdot a \subset h} \pi_{\mathcal{P}(h^{'})}(h^{'},a)$ is player $i$'s contribution. Given joint strategy $\pi=(\pi_1, \pi_2)$, we can define the \textbf{expected value} of history $h$ of player $i$, denoted by $v_i(h)$. If reaching probability $x^{\pi}(h)=0$, then $v_i^\pi(h)=0$; Otherwise, we have
\begin{align}
\label{ev}
v_i^\pi(h) =
    \sum_{z\in Z} \frac{x^{\pi}(z)}{x^{\pi}(h)} v_i(z),\ x^{\pi}(h)>0 ,
\end{align}
and thus the value function is \emph{linear} with respect to strategy $\pi$:
\begin{align}
v_i(\pi_i, \pi_{-i}) = \sum_{z \in Z} v_i^{\pi}(z)x^\pi(z).
\end{align} 

Then \textbf{best response} $\mathbb{BR}_i(\pi_{-i})$ is $\arg\max_{\pi_i}v_i(\pi_i, \pi_{-i})$. $\epsilon$-\textbf{Nash Equilibrium (NE)} strategy $\pi^*$ satisfies that $\forall i\in \mathcal{P}$, $\min_{\pi_{-i}} v_i(\pi^*_i, \pi_{-i}) + \epsilon \geq v_i(\pi^*)\geq \max_{\pi_i} v_i(\pi_i, \pi^*_{-i}) - \epsilon$. In particular, exact NE is $\epsilon$-NE when $\epsilon=0$. \textbf{Exploitability} $e(\pi) = \sum_{i\in \mathcal{P}} v_i(\mathbb{BR}(\pi_{-i}), \pi_i) - v_i(\pi)$. \textbf{Support Size} of NE ($\pi^*$) denotes the number of actions that have positive probabilities at infoset $s$ in NE $\pi^*$, denoted by $\text{supp}^{\pi^*}(s)$ ~\citep{supbound}. 

\subsection{Regret Minimization Algorithms}
Regret minimization methods approximate NE in EFGs if the algorithm has a sublinear regret upper bound or an average regret converging to zero as iteration $T$ goes to infinity. In this context, we introduce different regret minimization algorithms, and we will demonstrate how these algorithms ensure convergence to NE.

\begin{definition} [\textbf{Regret}]
Given $\{\pi^t|\ t=1,\cdots,T\}$ is a sequence of strategies delivered by an algorithm, the regret of this algorithm is defined as: 
\begin{align}
    R^T_i 
    = \sum_{t=1}^T \max_{\pi} v_i(\pi, \pi_{-i}^t) - v_i(\pi_i^t, \pi_{-i}^t),
\end{align}
and \textit{average} regret $\bar{R}^T_i = R^T_i/T$.
\end{definition}
Vanilla Counterfactual Regret Minimization (CFR) ~\cite{cfr} is a regret minimization algorithm that aims to minimize counterfactual regret at each infoset by traversing the full game tree depth-firstly. It achieves this by calculating player i's expected values $v_i(\cdot)$ based on equation \ref{ev} and computing instantaneous regrets at iteration $t\leq T$ of taking action $a$ in infoset $s$ following $r_i^t(s,a) = \sum_{h \in s}x_{-i}^{\pi^t} (h) [v_i^{\pi^t}(h\cdot a) -  v_i^{\pi^t}(h)]$. The counterfactual regrets of CFR can be computed by uniform average of $r_i$ over all iterations: $R_i^T(s,a) = \sum_{t=1}^T r_i^t(s,a) / T$. In two-player games, if both players apply regret matching~\citep{cfr} to strategy updates, the regret of CFR is bounded by$\Delta|S_i|\sqrt{|A_i|} \cdot T^{1/2}$.

% \textbf{Stochastic Regret Minimization}
% Rather than computing exact regret via entire tree traversal, Stochastic Regret Minimization samples nodes to traverse and estimate regret. In this family, we introduce Outcome sampling MCCFR that has the regret bound $(1/\delta + \sqrt{2}/\delta\sqrt{p})\Delta|S_i|\sqrt{|A|}\cdot T^{1/2}$, for any $p\in(0,1]$ and exploration parameter $\delta>0$.

Discounted regret minimization methods aim to minimize a weighted-average regret in which the strategy in each iteration is discounted:
\begin{align}
    \bar{R}^T_i = \sum_{t=1}^T w_t \left[\max_{\pi} v_i(\pi, \pi_{-i}^t) - v_i(\pi_i^t, \pi_{-i}^t)\right] / \sum_{t=1}^T w_t.
\end{align}
The discounted CFR (DCFR) ~\citep{discfr} is a regret minimization framework that belongs to this family and is based on the counterfactual regret minimization algorithm. DCFR employs weighted-average counterfactual regrets and weighted-average strategy to achieve faster convergence. It has an upper bound on weighted-average regret, where $w_t$ satisfy $\sum_{t=1}^\infty w_t=\infty$. The weighted-average regret's upper bound is $6\Delta|S_i|(\sqrt{|A|} + \frac{1}{T}) \cdot T^{-1/2}$. DCFR can be generalized to other CFR variants such as vanilla CFR, CFR+, and Linear CFR with appropriate hyperparameters ~\citep{cfrplus,discfr,noamthesis}.

For the purpose of facilitating analysis, we adopt a uniform notation for the upper bound of regret throughout the remaining sections of this paper. Specifically, we use $\mathcal{O}(|S_i|\sqrt{|A_i|}T^{-1/2})$ as the \textit{average} regret bound of CFR algorithms including vanilla CFR and DCFR.

\subsubsection{Regret Minimization to Nash Equilibrium}
A widely accepted folk lemma suggests that if both players in a two-player zero-sum game adopt an algorithm with a sublinear regret bound, then the average strategies of both players will converge to a Nash Equilibrium ~\citep{cesa2006prediction,blum2007learning}. Here we prove the discounted regret version with the same idea.

\begin{lemma}
\label{lemma:bound_to_ne}
Given the weighted-average regret of an algorithm $\mathcal{O}(|S_i|\sqrt{|A_i|}/\sqrt{T})$, the weighted-average strategy of this algorithm is a $\mathcal{O}(|S|\sqrt{|A|})/\sqrt{T})$-NE.
\end{lemma}

The proof is in Appendix \ref{append:bound_to_ne}. Therefore, if the weighted-average regret can converge to zero, the resulting weighted-average strategy is a reliable approximation that converges to NE.

\subsection{Double Oracle Methods}\label{section: Double Oracle}

\begin{algorithm}[t!]
    \caption{XDO~\citep{xdo}}
    \label{alg:xdo}
    \begin{algorithmic}
    \STATE Input: initial threshold $\epsilon_0$, time window index $j=0$, uniform random strategy $\pi^0$.
    \STATE $\Pi_1= \cup_{i \in \{1,2\}} \mathbf{BR}_i(\pi^0_{-i})$.
    \STATE Construct restricted game $\mathbf{G}_1$ with $\Pi_1$.
    \FOR{$t=1,\cdots,\infty$}
    \STATE Run one iteration of CFR in $\mathbf{G}_t$.
    \IF{Exploitability in $\mathbf{G}_t$, $e(\bar{\pi}^t)\leq\epsilon_0/2^j$}
    \STATE $\Pi_{t+1} = \Pi_t \cup \mathbf{BR}(\bar{\pi}^t_{-i})$ for $i \in \{1,2\}$.
    \STATE $j = j + 1$.
    \STATE Reset $\pi^{t+1}$.
    \STATE Construct restricted game $\mathbf{G}_{t+1}$ with $\Pi_{t+1}$.
    \ENDIF
    \ENDFOR
    \end{algorithmic}
\end{algorithm}

% \begin{figure}[t!]
%   \centering
%   \includegraphics[width=.236\textwidth]{figures/tree1_new.pdf}
%   \includegraphics[width=.236\textwidth]{figures/tree2_new.pdf}
% \caption{An example of Extensive-form Double Oracle. Square player is unaware of what Circle player chose. The values at the leafs are Circle player's payoff, since it is zero-sum game, the Square player's payoffs are the opposite numbers; At the beginning in Iteration 1, both Circle and Square players have empty population and uniform random strategy. In Iteration 2, Circle player's BR to Square's random strategy is action $L$ since $L$ and $R$'s expected values are $1.5$ and $1$ respectively. Square player's BR is action $Y$ since $X$ and $Y$'s EVs are $-2.5$ and $0$. Then $L$ and $Y$ are added to the population (thicker lines). Then since both players only have one action in the restricted game, the restricted game's NE is to choose them. Circle player's BR to NE of restricted game is still action $L$ given Square player will choose $Y$. Meanwhile, Square player's BR is still $Y$ given Circle player will choose $L$, $v(X)=-2$ and $v(Y)=-1$. Since there is no BR actions added to the population, the restricted game won't change in latter iterations, DO converges and the NE for the original game is $(L,Y)$.}
% \label{fig:example}
% \end{figure}

Double Oracle (DO) ~\citep{do} is a technique initially developed for solving Normal-form Games (NFGs) which maintains a population of pure strategies for both players, denoted as $\Pi_t$ at time $t$, and creates a restricted game by considering only the actions in $\Pi_t$. Then the restricted game's NE is obtained using linear programming, and the best response to this NE is added to the population. This process repeats until the population no longer changes. DO's advantage is that it solves large games by only solving some small restricted games, since the restricted game usually stops growing when it is still small, particularly when the support of the NE is small (i.e., non-zero probabilities of NE strategy's actions are few)~\cite{do_small_sub1, do_small_sub2}. Table 2 in Appendix \ref{table:support} includes the support size of NE in some common games. However, traditional algorithms such as linear programming can still be intractable in large restricted games. To address this issue, Online Double Oracle (ODO) ~\citep{odo} uses regret minimization for strategy updates in the restricted games, leading to better empirical performance than DO in large normal-form games.

The conversion of extensive-form games (EFGs) into normal-form games and solving them with normal-form Double Oracle (DO) is theoretically feasible. However, representing EFGs in normal-form will result in an exponential increase in the number of required iterations for convergence ~\cite{xdo}. Therefore, a specific DO algorithm for EFGs is necessary. The first DO method for EFGs, Sequence-form DO (SDO), was proposed by \citet{sdo}. This approach uses the Sequence-form LP to compute the exact NE of the restricted game. Moreover, SDO introduces an efficient technique for expanding the restricted game in extensive-form algorithms, where each iteration involves adding sequences instead of one pure strategy to the population. To address the challenge of solving large restricted games, \citet{xdo} extends SDO to Extensive-Form DO (XDO), which employs CFR to approximate the restricted game's NE and reinforcement learning methods to approximate the best response. Since XDO adds actions at every information set, it expands the restricted game much faster than SDO. The formal process of XDO is in Algorithm \ref{alg:xdo}, which is rearranged to count one iteration as the completion of one strategy update in the restricted game, followed by computing a BR if required. For more details on Algorithm \ref{alg:xdo}, please refer to Appendix \ref{append:algo_xdo}.

% First XDO starts with a population $\Pi_t$ and construct restricted game by restricting the available pure strategies in this population. Then XDO traverses the restricted game tree and update strategy with CFR. If the average strategy $\bar{\pi}$ reach the exploitability threshold $\epsilon/2^j$ in $\mathbf{G}_t$, XDO will expand $\Pi_t$ with Best Response to the restricted game's NE and increase j by $1$. XDO will repeat these steps until when $\bar{\pi}$ reach the $\epsilon/2^j$-NE and meanwhile $\mathbf{BR}(\bar{\pi})$ is in $\Pi_t$. At this time, XDO converges and reaches a global $\epsilon/2^j$-NE. 

To demonstrate the efficacy of DO for EFGs, we present a basic example of XDO (SDO) in a two-player zero-sum EFG, as depicted in Figure \ref{fig:example}. DO methods can compute the NE in EFGs with a small support NE without solving the original game. XDO has been empirically shown to converge rapidly in small support games~\citep{xdo}, but it lacks a theoretical analysis of the convergence rate for achieving an approximate NE. In the following section, we introduce a general framework that can generalize to both XDO and extensive-form ODO, and provide a theoretical analysis of their convergence rates. Based on the analysis, we propose a more sample-efficient DO algorithm for EFGs.

\section{Regret-Minimizing Double Oracle}
In this section, we propose Regret-Minimizing Double Oracle (RMDO), a novel generic Double Oracle framework combined with regret minimization to approximate the Nash Equilibrium of EFGs. To the best of our knowledge, this is the first study analyzing the convergence rate and sample complexity of regret-minimization based Double Oracle for EFGs.

RMDO consists of the same elements as the previous DO methods. Restricted game is constructed by considering only a subset of all pure strategies. Population $\Pi_t$ containing the available pure strategies in the restricted game. Time window $T_j$, defined as a partition of the set of all iterations where the populations are the same: $\forall t_0,t_1 \in T_j, \Pi_{t_0}=\Pi_{t_1}$, plays a crucial role in RMDO and contributes to making it a generic framework. The number of time windows, denoted by $k$, corresponds to the number of restricted games from iteration $t=0$ to $T$. However, in contrast to existing DO methods, RMDO has the ability to expand the restricted game at any time.

Prior to presenting the formal process, we highlight two key new components of RMDO. The first component is the \textbf{frequency function} $m(\cdot)$ used in the computation of Best Response, denoted as $m(\cdot)$, which is defined as a mapping from the set of time window indices $\mathcal{N} \cap [0,k-1]$ to $\mathcal{N}^{+}$. The function $m(j)$ represents the frequency of computing Best Response in the $j$-th time window. Since the process of DO based on regret minimization is exactly to take turns to do regret minimization and compute the best response, balancing the regret minimization and Best Response computation is critical to achieve a rapid convergence. The second component is \textbf{weighted-average scheme}. To accelerate the convergence in the restricted game, we incorporate within-window weights $w_t$, where $t\in T_j$, allowing us to utilize the discounted regret minimizer. The within-window weights are exactly the weights in the discounted regret minimizer, satisfying that their sum in the current window $T_j$ equals one. For instance, vanilla CFR employs uniform weights $1/|T_j|$ in window $T_j$. CFR$+$ and linear CFR use linearly increasing weights. Specifically, within the $j$-th window $T_j$:
\begin{align}
w_t= 2(t - \sum_{m=1}^{j-1} |T_m|)/|T_j|(|T_j| + 1). 
\end{align}
Thus, the weighted average strategy in the window $T_j$ is:
\begin{align}
\label{eq:avg_strategy}
\Tilde{\pi}_i^t = \sum_{t'\in T_j}\pi_i^{t'} w_{t'}.
\end{align}

Presented in Algorithm \ref{alg:RMDO}, the formal RMDO procedure is as follows. At each iteration $t$, assuming the current time window is $j$, the restricted game $\mathbf{G}_t$ is constructed by restricting the pure strategies in the population $\Pi_t$ for players. Within $\mathbf{G}_t$, regret minimization is conducted by traversing the game tree, computing the regret of each infoset (node), and updating the strategy using any Counterfactual Regret Minimization (CFR) algorithm. At the outset of the procedure, when $t=0$, the construction of the restricted game and the strategy update are bypassed since $\Pi_0$ is empty. The expected value at $t=0$ is computed based on the joint strategy $\pi$ following a uniformly random policy. As the procedure progresses, when $t>0$ and the current time window is $T_j$, the joint average strategy of current window $\bar{\pi}=(\bar{\pi}_1, \bar{\pi}_2)$ is expanded to the original game every $m(j)$ iteration by setting the probabilities of actions not in the population to zero. Then the original game best response (BR), considering all actions in the original game, is computed against the expanded current-window average strategy, which is $\mathbf{a}_i^t=\arg\max_{\pi_i \in \Pi}v_i(\pi_i, \pi_{-i})$, for both players. $\mathbf{a}_i^t$ for $i=1,2$ are both merged to the population $\Pi_{t+1}$. Finally, if the population changes ($\Pi_{t+1} \neq \Pi_t$), a new time window is initiated, and $\pi_i^{t+1}$ is reset to a uniform random strategy.

\begin{algorithm}[t!]
\begin{algorithmic}
    \caption{Regret-Minimizing Double Oracle}
    \label{alg:RMDO}
    \STATE Input: hyperparameter $m$, window index $j=0$, uniform random strategy $\pi^0$.
    \STATE Set population $\Pi_1=\mathbf{BR}_i(\pi^0)$ for $i \in \{1,2\}$.
    \STATE Construct restricted game $\mathbf{G}_1$ with $\Pi_1$.
    \FOR{$t=1,\cdots,\infty$}
    \STATE Run one iteration of CFR in $\mathbf{G}_t$.
    \IF{$t \mod m(j)=0$}
    \STATE Compute $\Tilde{\pi}_i^t$ with equation (\ref{eq:avg_strategy}).
    \STATE
    $\Pi_{t+1} = \Pi_t \cup  \mathbf{BR}_i(\Tilde{\pi}^t_{-i})$ for $i \in \{1,2\}$.
    \IF{$\Pi_{t+1} \neq \Pi_{t}$}
    \STATE Start new window: $j = j+1$.
    \STATE Reset strategy $\pi^{t+1}$.
    \STATE Construct restricted game $\mathbf{G}_{t+1}$ with $\Pi_{t+1}$.
    \ENDIF
    \ENDIF
    \ENDFOR
\end{algorithmic}
\end{algorithm}

Then we investigate the convergence guarantee of RMDO. Regret minimization algorithms can converge to $\epsilon$-NE by iteratively  updating strategy in a static game. But in RMDO, a regret minimizer is employed in the restricted game, which expands over time. Thus if the restricted game stops expanding at some finite iteration, the convergence of RMDO is guaranteed. The following lemma proves that the number of restricted games is finite, which guarantees RMDO's convergence.

\begin{lemma}
\label{lemma:k}
In an extensive-form game, where $\Pi^*$ represents the set of Nash Equilibrium (NE) strategies of this game. When running RMDO, given that $k$ is the number of restricted games, we have $\min_{\pi\in \Pi^*}\max_{s\in S}\text{supp}^{\pi}(s) < k \leq \sum_{i}|S_i|$. (Refer to Appendix \ref{append:k} for proof).
\end{lemma}

RMDO will converge by doing regret minimization in the final restricted game, after $k$ times of restricted game expanding. Additionally, according to Lemma \ref{lemma:k}, where $k$ is bounded, the subsequent assumption can be made in hindsight, without any loss of generality, for the remainder of this paper.
\begin{assumption}
\label{assume:k}
Given an extensive-form game solving by Regret-Minimizing Double Oracle, there are $k<\infty$ restricted games even when $T\rightarrow \infty$.
\end{assumption}

While the value of $k$ is unknown during training, it is possible to partition the iterations into a limited number of time windows in hindsight to investigate the convergence rate and sample complexity of RMDO. In the following section, we will examine two types of strategies that are generated by RMDO.

\subsection{Overall Average Strategy}
The convergence rate of the average strategy can be determined by the regret upper bound of the regret minimization algorithm. Thus we first investigate the overall average strategy by taking the average over iterations $t=0$ to $T$.

In addition to the within-window weights, global weights are assigned to $\pi_t$ when computing the overall average strategy. Specifically, for all $t\in T_j$ within a time window $T_j$, the weight $W_t$  is defined as $|T_j|w_t / T$. It can be easily shown that $\sum_{t=1}^{T} W_t = 1$. The overall average strategy for player $i$ is then obtained as:
\begin{align}
    \bar{\pi}_{i}^t = \sum_{t=1}^T  \pi_{i}^{t} \cdot W_t = \sum_{t=1}^T  \pi_{i}^{t} \cdot \frac{|T_j|w_t}{T}.
\end{align}
Define weighted-average regret of RMDO as
\begin{align}
    \bar{R}^T_i = \max_{\pi^{'}_i} \sum_{t=1}^T (v_i(\pi^{'}_i, \pi_{-i}^t) - v_i(\pi^t)) \cdot W_t.
\end{align}

We then prove that the regret has the following upper bound:

\begin{theorem}
\label{theo:RMDO_bound} 
In RMDO, suppose the regret minimizer has $\mathcal{O}(|S_i|\sqrt{|A|T})$ regret, the weighted-average regret bound of RMDO:
\begin{align}
\mathcal{O}\left(\sum_{j=0}^{k-2} \frac{|T_j|}{T}\cdot[m(j)-1] + \sum_{j=0}^{k-1}\frac{\sqrt{k}|S_i||T_j|}{T\sqrt{|T_j| - m(j) + 1}}
    \right),
    \label{eq:rmdo}
\end{align}
converges to $0$ if $m(j)$ is sublinear in $T$.
\end{theorem}

The proof is in Appendix \ref{append:RMDO_bound}. If the number of restricted games $m(j)$ is sublinear in the total number of iterations $T$, then the Regret-Minimizing Double Oracle (RMDO) algorithm is an anytime algorithm that converges to Nash Equilibrium (NE) with its overall average strategy. Moreover, with the help of Lemma \ref{lemma:bound_to_ne}, one can easily obtain the expected number of iterations required for the algorithm to converge to $\epsilon$-NE.

% \begin{proposition}
% \label{theo:RMDO_cr} Assume $\mathcal{O}(\sum_j |T_j| [m(j)-1]) \sim \mathcal{O}(T^\alpha)$. Then to let the overall average strategy of RMDO reach $\epsilon$-NE, the sample complexity is $\mathcal{O}(\epsilon^{1/\alpha})$ when $1/2 \leq \alpha < 1$, is $\mathcal{O}(|S|^3k^2/\epsilon^2)$ when $\alpha \leq 1/2$. (Proof in Appendix \ref{append:RMDO_cr}.)
% \end{proposition}

\subsection{Last-Window Average Strategy}

The final average strategy obtained from regret minimization in the last time window is an $\epsilon$-NE, requiring at least $\mathcal{O}(|A_{i,k}||S_{i,k}|^2/\epsilon^2)$ iterations to reach, where $A_{i,k}$ and $S_{i,k}$ denote the action and information set space in the last time window $T_k$ of player $i$. The regret bound of the regret minimizer is $\mathcal{O}(|S_{i,k}|\sqrt{|A_{i,k}|/T})$. During the growth of the population, we observe that in each time window, the regret minimizer at each iteration except the last one will not reach $\epsilon$-NE. Otherwise, in a non-last time window, the global best response is already in the population. Thus, the average strategy in this window at this iteration will already reach $\epsilon$-NE, which is contradictory. Utilizing this idea, we can bound the number of iterations required for each time window and provide a sample complexity to reach $\epsilon$-NE.

\begin{theorem}
\label{theo:lw-xodo}
The last-window average strategy of RMDO needs the following number of iterations to reach $\epsilon$-NE:
\begin{align}
    \mathcal{O}(k|A||S|^2/\epsilon^2 -k + \sum_j m(j)).
\end{align}
\end{theorem}

The proof is in Appendix \ref{append:lw-xodo_it}. Utilizing this result, we can estimate the sample complexity required for RMDO to achieve $\epsilon$-NE. It is noteworthy that as the exact regret minimizer will traverse the entire game tree, the sample complexity is at most $\mathcal{O}(|S|)$. Hence, the sample complexity for the regret minimization part is given as $\mathcal{O}(k|A||S|^3/\epsilon^2 -k|S| + |S|\sum_j m(j))$.

In addition, we also need to investigate the sample complexity involved in computing the best response (BR). As BR computation also requires full tree traversal, we only need to consider the number of times RMDO computes BR during training. Since it is reliant on the selection of frequency function $m(j)$, we will examine the overall sample complexity in the following section where we introduce RMDO with various schemes of frequency functions and demonstrate how it generalizes to existing methods. Following this analysis, we propose a more sample-efficient algorithm in comparison to existing methods.

\section{Efficient Schemes of Frequency Function}
The complexity of approximating NE using RMDO is influenced by the choice of frequency function $m(j)$ for best response computation, as demonstrated in the previous section through theoretical analysis. For analysis on existing methods, we present various RMDO instantiations with distinct frequency schemes in this section.

\subsection{Online Double Oracle for Extensive-Form Games}

\begin{algorithm}[t]
    \caption{Extensive-Form Online Double Oracle}
    \label{alg:xodo}
\begin{algorithmic}
    \STATE Initial empty population $\Pi_0$, time window $T_0$
    \FOR{$t=0,\cdots,\infty$}
    \STATE Construct restricted game $\mathbf{G}_t$ with $\Pi_t$.
    \STATE Update strategy $\pi^t$ in $\mathbf{G}_t$.
    \FOR {$i = 1,2$}
    \STATE Get BR: $\mathbf{a}_i^t = \arg\max_{\pi_i \in \Pi} \sum_{t' \in T_j} v_i(\pi_i, \pi_{-i}^{t'}).$
    \STATE Update population $\Pi_{t+1} = \Pi_t \cup \mathbf{a}_i^t$.
    \ENDFOR
    \IF{$\Pi_t \neq \Pi_{t+1}$}
    \STATE Start new window $T_{j+1}$ and Reset strategy $\pi_i^{t+1}$.
    \ENDIF
    \ENDFOR
\end{algorithmic}
\end{algorithm}

We propose an extension of the Online Double Oracle (ODO) algorithm, known as the Extensive-Form Online Double Oracle (XODO), which combines the Sequence-form DO framework with Counterfactual Regret Minimization (CFR) to solve extensive-form games. The algorithm is described in detail in Algorithm \ref{alg:xodo}, where the construction and update of the restricted game and strategy are performed in a similar manner to the DO framework used in ODO. However, in each iteration, XODO expands the restricted game by computing the best response against the average strategy in the current window. As XODO computes the best response in each iteration after regret minimization, it is equivalent to the RMDO algorithm with $m(j)=1$. Based on Theorem \ref{theo:RMDO_bound}, XODO has a regret bound.

\begin{corollary}
\label{theo:xodo_bound} 
In XODO, given the regret minimizer with $\mathcal{O}(|S_i|\sqrt{|A|T})$ regret upper bound, the weighted-average regret bound of XODO will be:
\begin{align}
    \mathcal{O}( \frac{|S_i|k}{\sqrt{T}}).
\end{align}
\end{corollary}
\begin{proof}
Plug $m(j)=1$ into equation \ref{eq:rmdo}, the upper bound will be $\mathcal{O}(|S_i|\sqrt{k}\sum_j\sqrt{|T_j|}/T)$. According to Cauchy-Schwartz inequality, $\sum_j\sqrt{|T_j|}\leq \sqrt{k\sum_j |T_j|}=\sqrt{kT}$, then the upper bound becomes $\mathcal{O}(|S_i|k/\sqrt{T})$.
\end{proof}

Then according to lemma \ref{lemma:k}, $k\leq |S|$, XODO has a sublinear regret upper bound. According to the regret to strategy conversion in Lemma \ref{lemma:bound_to_ne}, the overall average strategy of XODO requires $\mathcal{O}(|S|^2k^2/\epsilon^2)$ iterations to reach $\epsilon$-NE. We can further derive the sample complexity:

\begin{proposition}
\label{theo:xodo_sc}
Since XODO compute BR in each iteration, the sample complexity to reach $\epsilon$-NE is $\mathcal{O}(2|S|^3k^2/\epsilon^2)$. (Proof in Appendix \ref{append:xodo_sc}).
\end{proposition}

\subsection{Extensive-Form Double Oracle}

The Extensive-form Double Oracle (XDO) algorithm is initialized with a given threshold $\epsilon_0$, which is divided by two each time the local exploitability of the regret minimizer meets the threshold. The local exploitability is the exploitability in the restricted game. In time window $T_j$, the algorithm performs regret minimization for more than $4^j|S_{i,j}|^2|A_{i,j}|/\epsilon_0^2$ iterations before computing the best response. Here $A_{i,j}$ and $S_{i,j}$ denote the action space and infoset space in the $j$-th time window of player $i$. Finally the average strategy in the last window is outputted when the convergence condition is met. If XDO converges, the last-window average strategy is $\epsilon_0/2^k$-NE. To investigate the complexity of reaching $\epsilon$-NE, it is assumed without loss of generality that $\epsilon_0/2^k\leq \epsilon$.

Thus, RMDO can generalize to XDO with $m(j)\geq 4^j|S_{i,j}|^2|A_{i,j}|/\epsilon_0^2$ and the last-window average strategy. Based on Theorem \ref{theo:lw-xodo}, we can determine the expected iterations and sample complexity for XDO to reach $\epsilon$-NE.

\begin{corollary}
\label{theo:xdo_cr}
XDO needs at least $\mathcal{O}(k|A||S|^2/\epsilon^2 + |A||S|^2 4^{k}/\epsilon_0^2-k)$ iterations to reach $\epsilon$-NE. 
\end{corollary}

\begin{proposition}
\label{theo:xdo_sc}
Since XDO computes BR only at the end of the time window before convergence, the sample complexity to reach $\epsilon$-NE is at least $\mathcal{O}(k|A||S|^3/\epsilon^2 + |A||S|^3 4^k/\epsilon_0^2)$.
\end{proposition}

Corollary \ref{theo:xdo_cr} is a specific instance of Theorem \ref{theo:lw-xodo} with an appropriate choice of $m(j)$, and its proof is provided in Appendix \ref{append:gxdo_sc}. The sample complexity of XDO is analyzed in Proposition \ref{theo:xdo_sc}, and its proof can be found in Appendix \ref{append:xdo_sc}. Theorem \ref{lemma:k} states that $k\leq |S|$; thus, theoretically, the restricted game stopping condition of XDO decays exponentially, implying that in the worst-case scenario, when $k=|S|$, XDO has an exponential sample complexity in the number of infosests. Therefore, XDO suffers from a large theoretical sample complexity. Empirically, the values of $k$ when executing XDO on common poker games leads to a large sample complexity (Appendix \ref{append:xdo}).

\begin{figure*}[t!]
     \centering
\begin{subfigure}
    \centering
    \includegraphics[width=.32\textwidth]{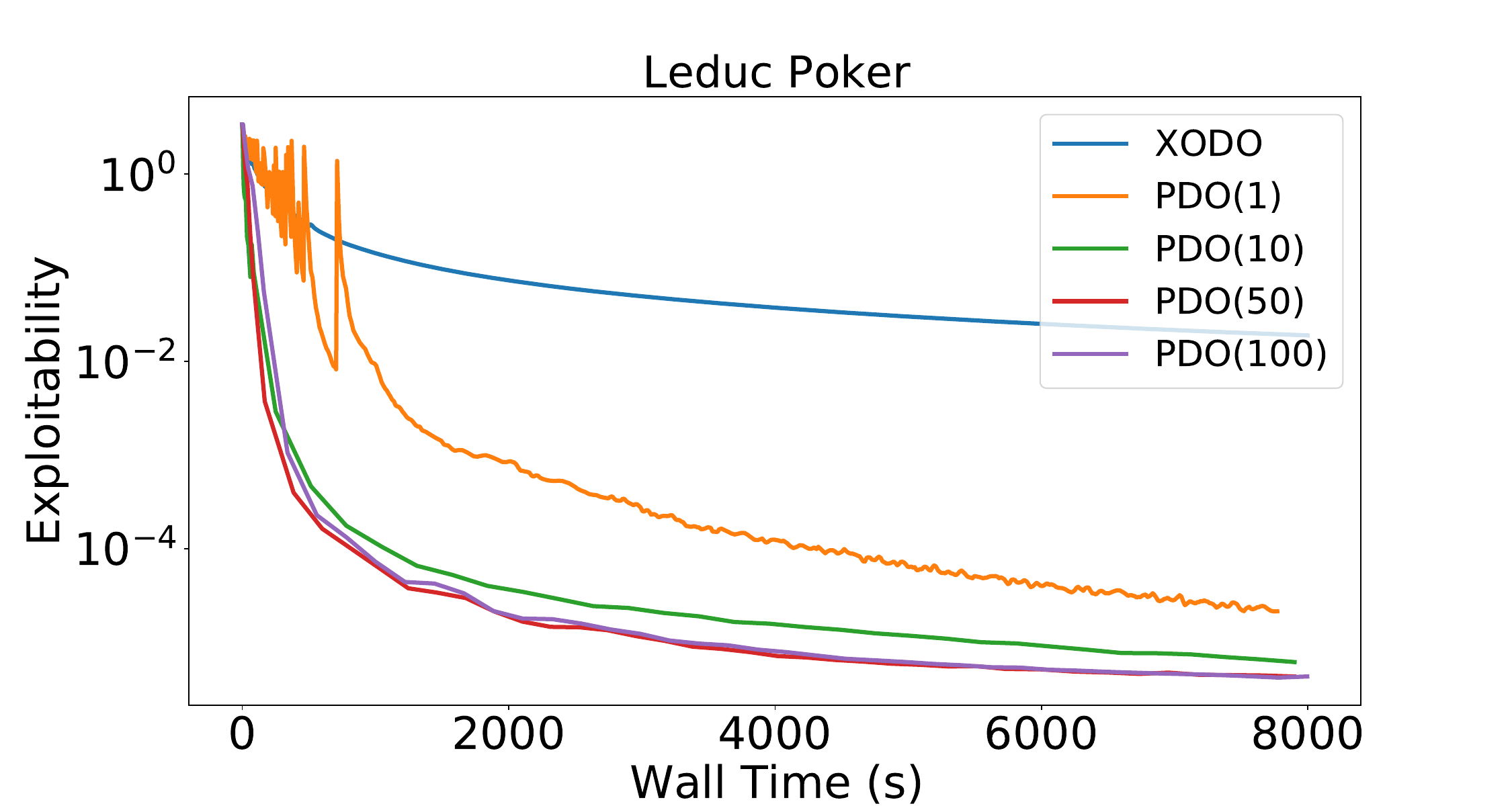}     
\end{subfigure}
\begin{subfigure}
    \centering
    \includegraphics[width=.32\textwidth]{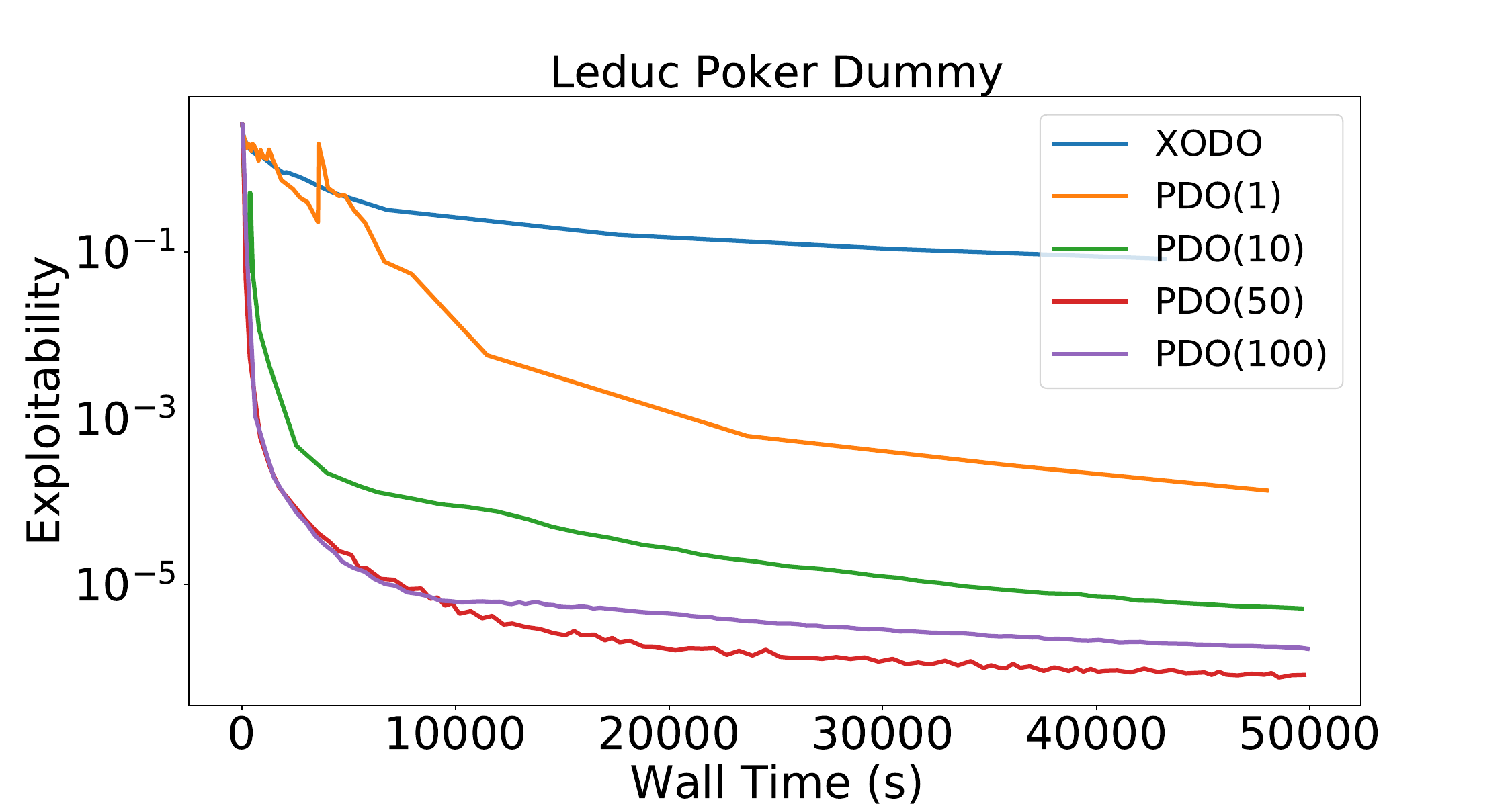}
\end{subfigure}
\begin{subfigure}
    \centering
    \includegraphics[width=.32\textwidth]{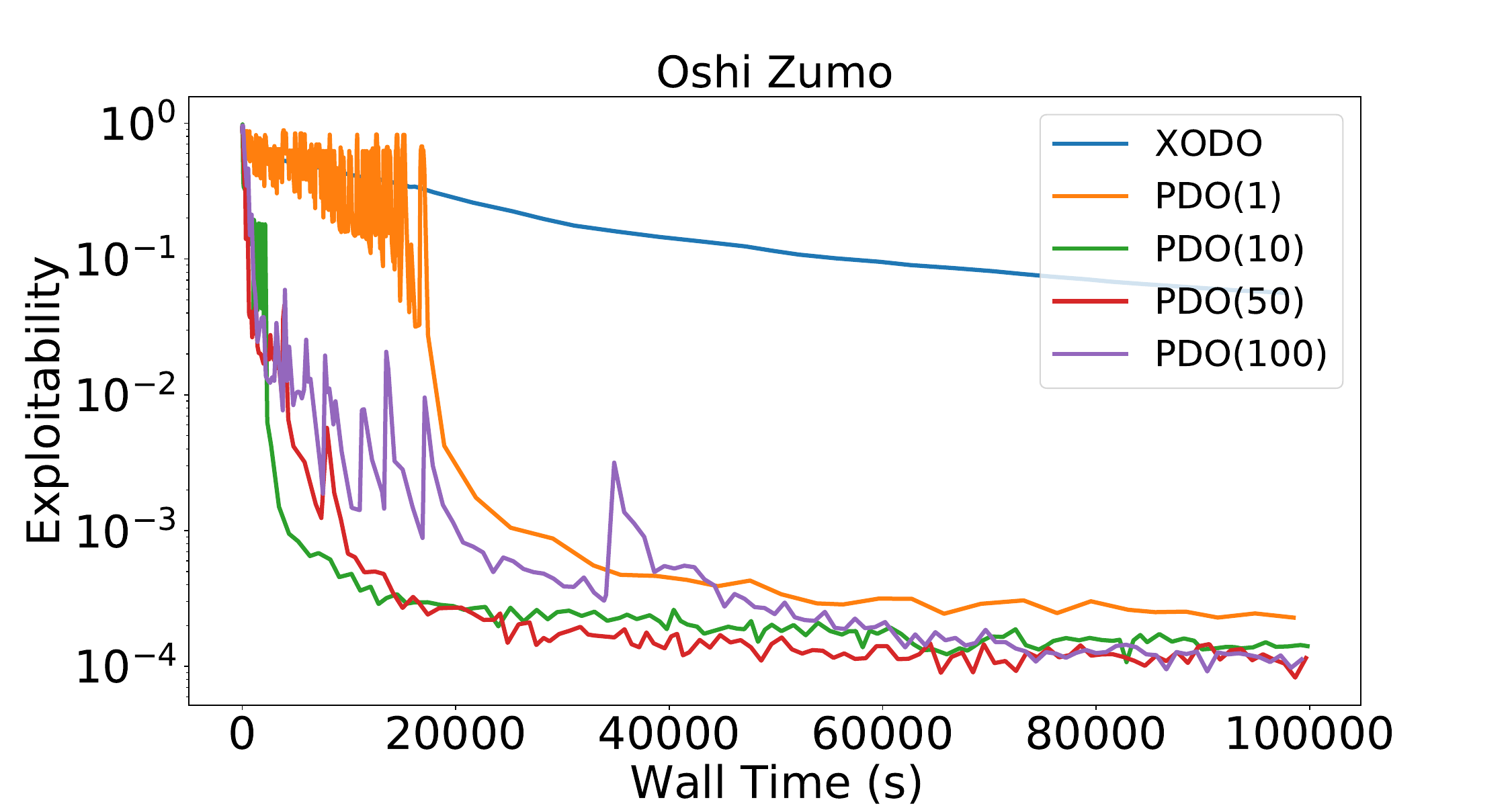}     
\end{subfigure}
\caption{Performance of XODO and PDO with periodicity function $m(j)=1,10,50,100$ on Leduc Poker, Dummy Leduc Poker and Oshi Zumo.}
  \label{fig:exploitability_ablation}
\end{figure*}

\begin{figure*}[t!]
     \centering
\begin{subfigure}
    \centering
    \includegraphics[width=.32\textwidth]{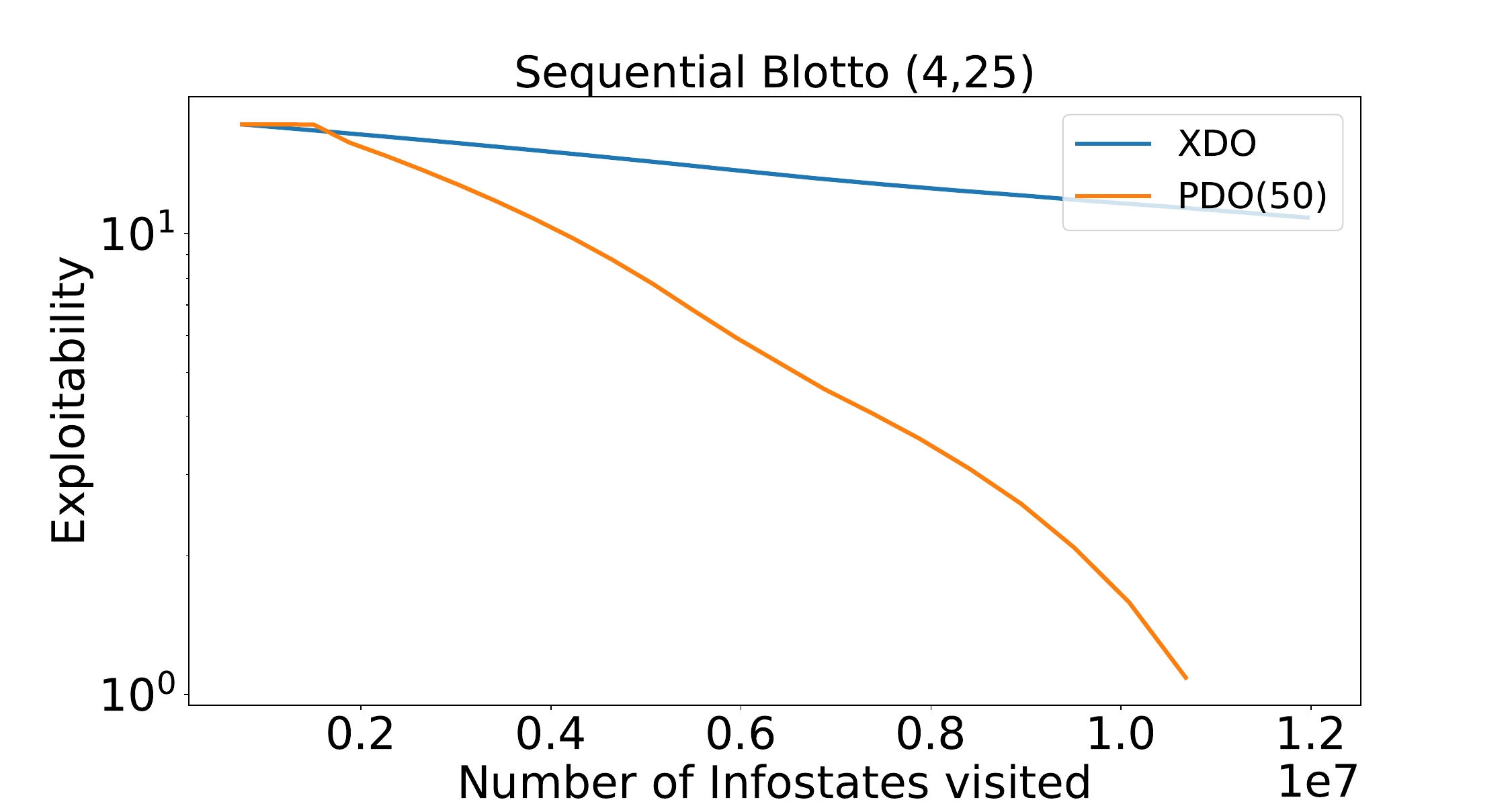}     
\end{subfigure}
\begin{subfigure}
    \centering
    \includegraphics[width=.32\textwidth]{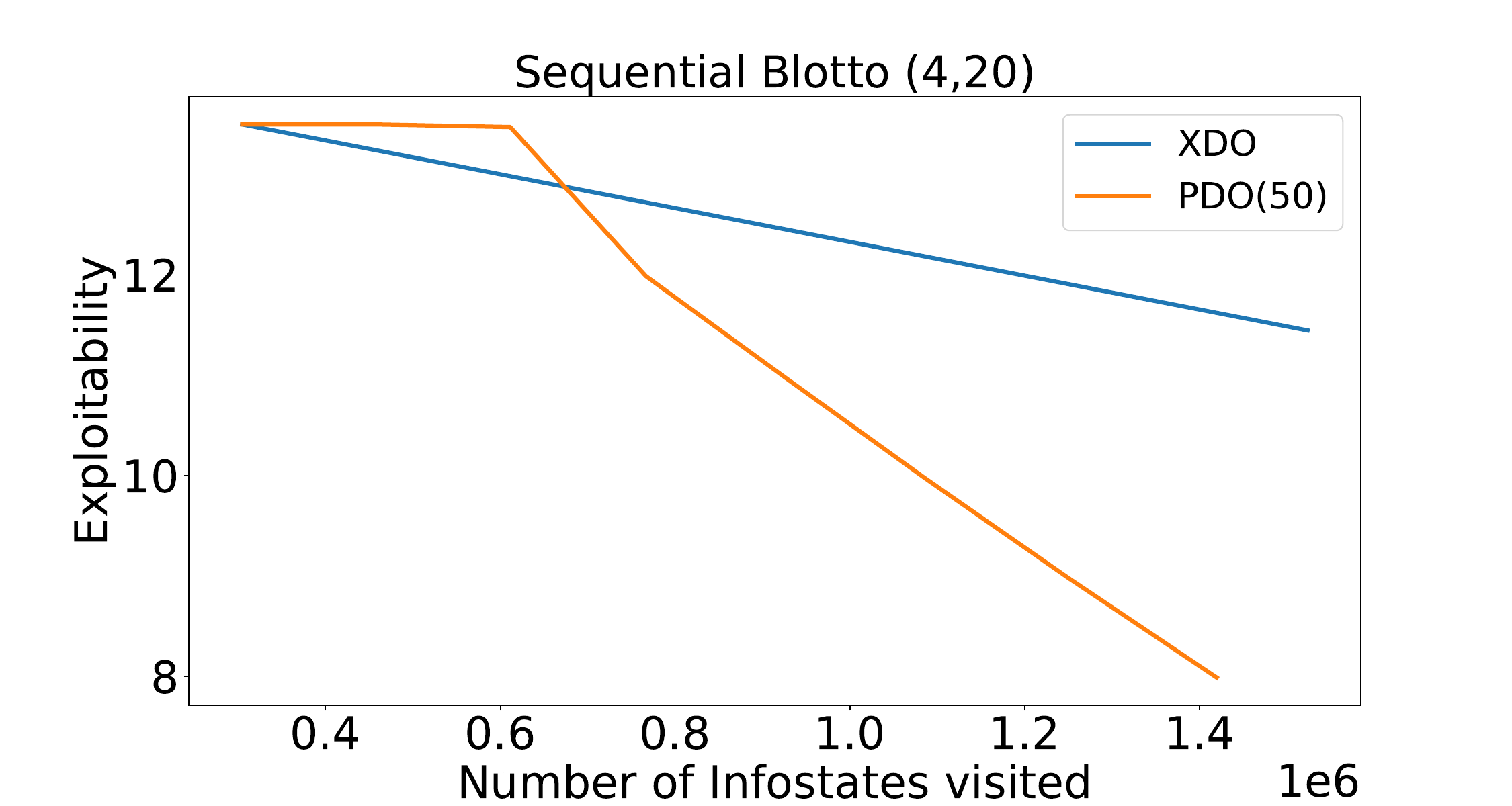}
\end{subfigure}
\begin{subfigure}
    \centering
    \includegraphics[width=.32\textwidth]{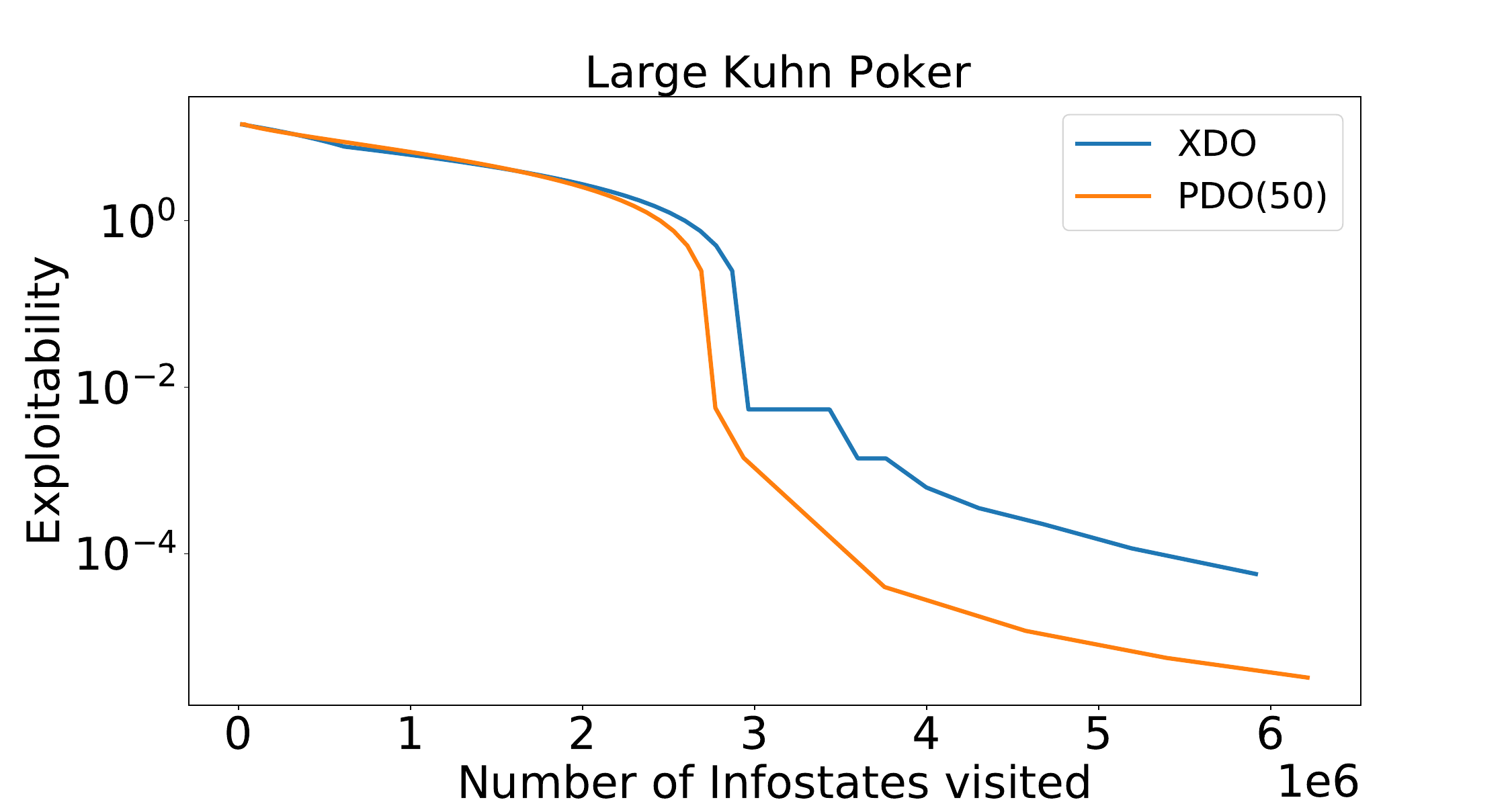}     
\end{subfigure}
\caption{Performance of PDO with periodicity $50$ and XDO on Sequential Blotto Games and Kuhn Poker with initial pot $40$ for each player. In Sequential Blotto, XDO is still significantly more exploitable than PDO even after visiting more than $10^6$ infosets.}
  \label{fig:exploitability_pdo_xdo}
\end{figure*}

\subsection{Periodic Double Oracle}
The exponentially growing frequency function $m(j)$ of XDO leads to an exponential increase in sample complexity with respect to $k$. On the other hand, XODO's inflexibility arises from the fact that it performs best response computation in each iteration, neglecting the balance between regret minimization and best response computation. In order to mitigate the large increase in sample complexity caused by a large value of $k$, and to balance the two computations, we propose the Periodic Double Oracle (PDO) algorithm. PDO computes the best response at a fixed time interval in each window and outputs the average strategy in the last window.

By setting $m(j)$ to a constant $c$, PDO can be viewed as an instantiation of RMDO. We can derive the expected number of iterations required for PDO to reach an approximate NE by utilizing Theorem \ref{theo:lw-xodo}. The algorithm is presented in Appendix \ref{append:pdo_add}.

\begin{corollary}
PDO needs $\mathcal{O}(k|A||S|^2/\epsilon^2 + (c-1)k)$ iterations to reach $\epsilon$-NE.
\end{corollary}

\begin{proposition}
\label{theo:pdo_sc}
Since PDO computes BR every $c$ iterations, the sample complexity to reach $\epsilon$-NE is $\mathcal{O}(k|A||S|^3/\epsilon^2 + ck|S| + k|A||S|^3/c\epsilon^2 - k|S|/c)$.
\end{proposition}

Proposition \ref{theo:pdo_sc} has been demonstrated in Appendix \ref{append:pdo_sc}. The periodicity $m(j)=c$ in PDO reduces the impact of the dominating term $|S|^3$ in sample complexity, compared to that of XODO. Additionally, compared to XDO, PDO eliminates the term exponential in $k$ from the sample complexity. While XDO may have an \textit{exponential} sample complexity in $|S|$ in the worst case scenario, PDO only has polynomial complexity in $|S|$. Hence, theoretically, PDO is more sample-efficient than existing algorithms (refer to Table \ref{table:sc} for a summary of sample complexities).

Given that these complexities are computed in the worst-case scenario, and are computed to ensure that the algorithm reaches at most $\epsilon$-NE, we cannot determine the value of $c$ by merely solving for extreme values of sample complexity. Instead, we consider it as a hyperparameter and analyze the empirical performance of PDO with different $c$ in the next section.

\section{Experiments}
\label{sec:exp}

\begin{figure*}[t!]
     \centering
% \begin{subfigure}
%     \centering
%     \includegraphics[width=.49\textwidth]{figures/kuhn_poker_all_infostates.pdf}     
% \end{subfigure}
\begin{subfigure}
    \centering
    \includegraphics[width=.32\textwidth]{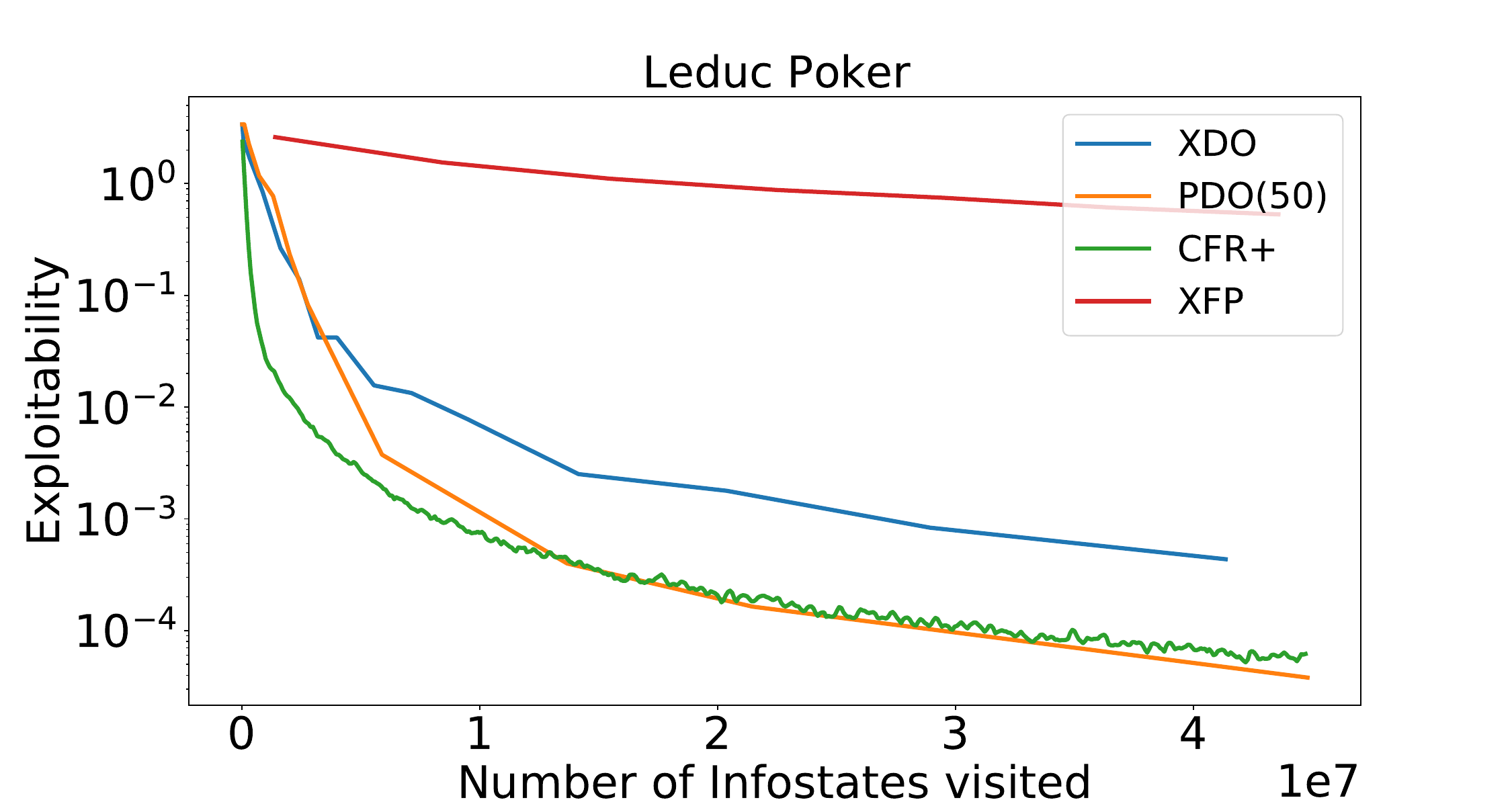}
\end{subfigure}
\begin{subfigure}
    \centering
    \includegraphics[width=.32\textwidth]{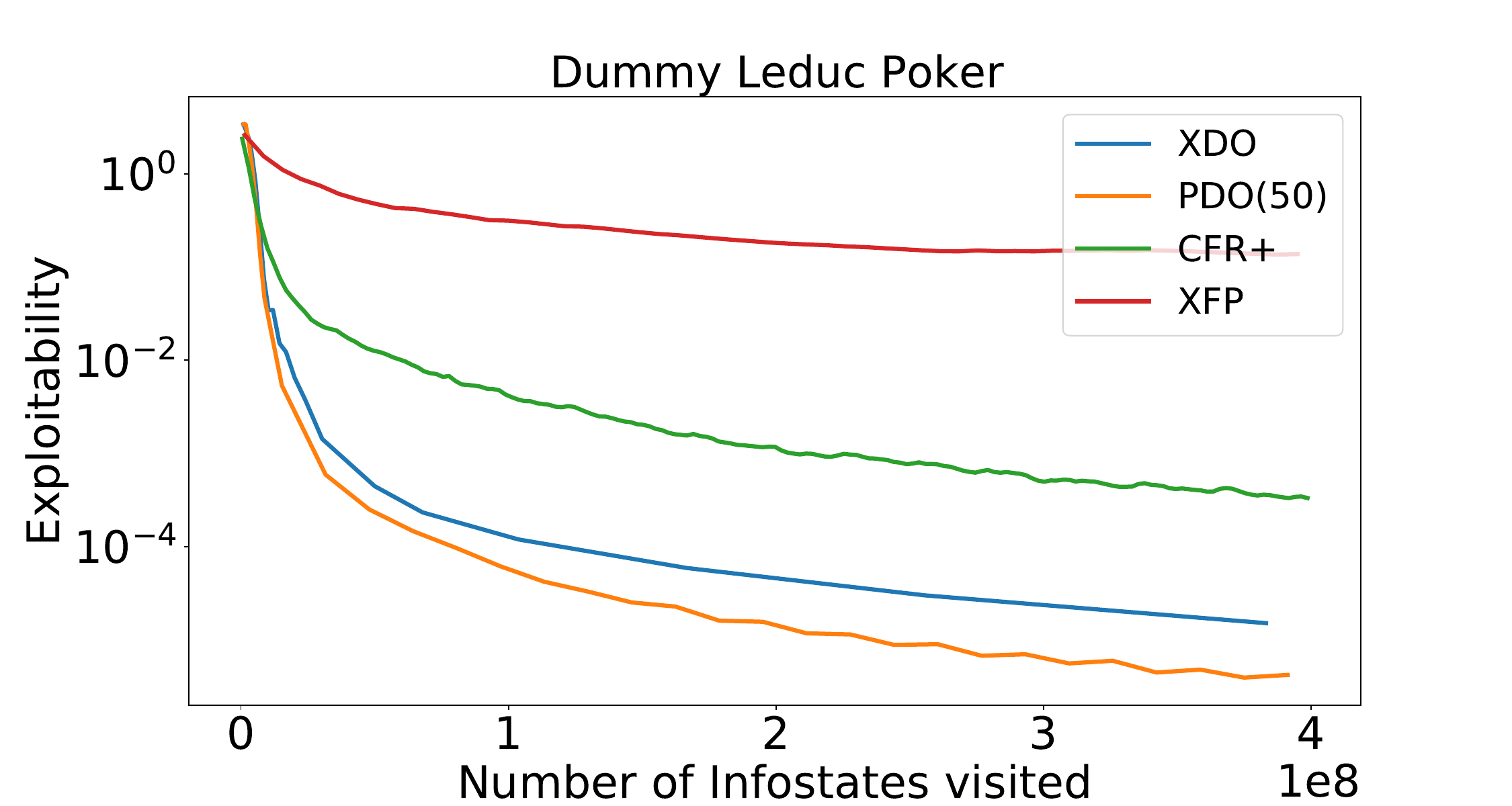}     
\end{subfigure}
\begin{subfigure}
    \centering
    \includegraphics[width=.32\textwidth]{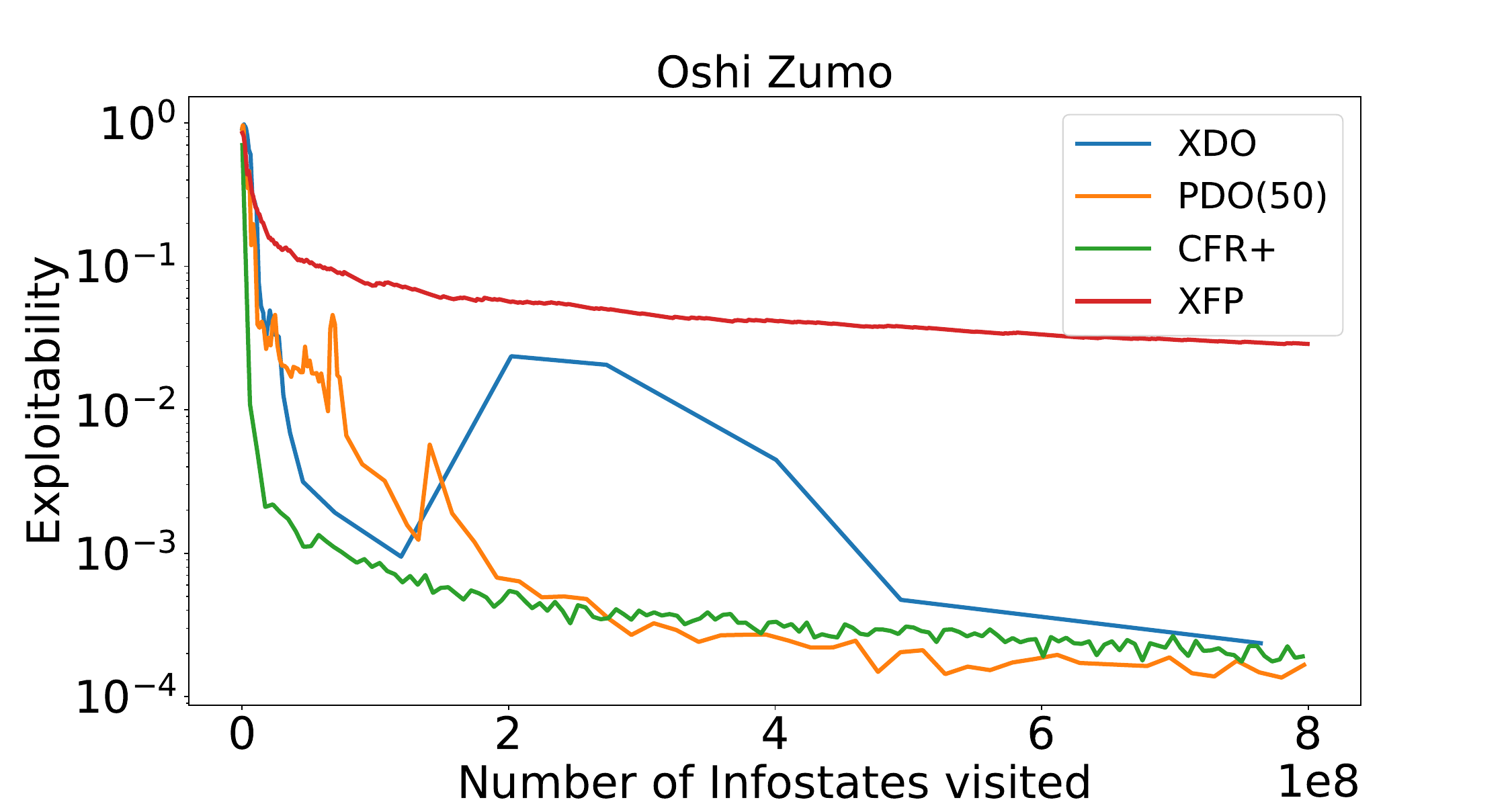}
\end{subfigure}
\caption{Comparison of PDO with XDO, CFR$+$ and XFP on Leduc Poker, Dummy Leduc Poker and Oshi Zumo. PDO preserves the strength of DO that performing well in games with small-support NE (Dummy Leduc Poker), but remains competitive to the state-of-the-art regret minimization methods in other games.}
  \label{fig:exploitability}
\end{figure*}

All experiments and algorithm implementations are based on OpenSpiel~\citep{lanctot2019openspiel}. Code is available at \href{https://github.com/xiaohangt/RMDO}{https://github.com/xiaohangt/RMDO}. We conducted empirical assessments on a variety of extensive-form games, including Sequential Blotto (perfect-information extensive-form game), Kuhn Poker with a initial pot $40$ for each player, Leduc Poker, Leduc Poker Dummy and Oshi Zumo. Leduc Poker Dummy is of particular interest as the NE of the game has a small support as the actions are duplicated in each infoset ~\citep{xdo}. Oshi Zumo is a board game in which players must repeatedly bid to push a token off the other side of the board ~\citep{oshizumo}. The full description of games is in Appendix \ref{append:games}.

We evaluate the performance with the exploitability in terms of number of infosets visited and wall time measured in seconds. The number of visited infosets refers to the total number of nodes traversed by the algorithm, including those encountered during best response (BR) computation. It is equivalent to the number of touched nodes in the experiment of Stochastic Regret Minimization ~\citep{srm}. It is important to note that the difference between the expanded infosets in XDO paper ~\citep{xdo} and the visited infosets in the present study is that XDO did not include the infosets during BR computation.

The experiment uses the state-of-the-art exact regret minimizer, CFR$+$ ~\cite{cfrplus}, for all double oracle algorithms, and the regret minimization algorithms are initialized with uniform random policies following the default setting. The study begins by analyzing the performance of PDO with various periodicity choices and comparing them with XODO. Furthermore, the algorithm's performance is compared against baselines, including XDO with restricted game solver CFR+, CFR+ itself, and Extensive-form Fictitious Self-play. For more empirical results including the comparison to Linear CFR, please refer to Appendix \ref{append:graphs}.

In Figure \ref{fig:exploitability_ablation}, we present a comparison between the XODO and PDO with periodicity values of $m=1, 10, 50, 100$, in terms of exploitability plotted against wall time in seconds.Our results show that in Kuhn poker, PDO algorithms outperform XODO, with all PDO algorithms exhibiting similar performance. Among the PDO algorithms, we found that PDO with periodicity $100$ performs slightly better than the other PDO algorithms. In all the other games, PDO algorithms outperformed XODO by a large margin. In Leduc Poker, larger periodicity values led to faster convergence. Among the PDO algorithms, PDO with periodicity values of $50$ and $100$ performed the best. In Leduc Poker Dummy, PDO with periodicity $50$ achieved a small exploitability the fastest. In Oshi Zumo, large periodicity values led to slow decreasing in exploitability at the early stages of training, but reached the least exploitability later on.

We also compare the performance of PDO with XDO in Figure \ref{fig:exploitability_pdo_xdo}. We find that PDO ($50$) outperforms XDO with a large margin in Sequential Blotto and Large Kuhn Poker. In Figure \ref{fig:exploitability}, we investigate the performance of PDO ($50$) and other baselines. We find that PDO outperforms XDO and Extensive-form Fictitious Self-play (XFP) with a large margin in Leduc Poker, and has a slight improvement over CFR$+$ in exploitability in the later stages of the training. In Leduc Poker Dummy, PDO outperforms all other algorithms with a large margin starting from the beginning of the training. In Oshi Zumo, PDO has a more stable exploitability curve and a faster convergence in general compared to XDO, with a lower exploitability level than CFR$+$ in the later stages of the training. 

Our findings suggest that PDO improves the convergence speed of DO methods in different types of games, and that the choice of periodicity can have a significant impact on the performance of PDO. Furthermore, we find that the last-window average strategy converges faster than the overall average strategy when comparing PDO ($1$) and XODO, likely due to the poor performance of the strategy before the restricted games stop expanding. These results contribute to the ongoing effort to improve the efficiency and effectiveness of DO algorithms in solving large-scale imperfect-information games.

% \begin{figure}
%     \centering
%     \includegraphics[width=.49\textwidth]{figures/kuhn_poker_s.pdf}
%     \caption{Performance of Stochastic XODO on Kuhn Poker}
%     \label{fig:SODO}
% \end{figure}

% We finally investigate the performance of Stochastic XODO on Kuhn Poker in Figure \ref{fig:SODO} via the curves of exploitability in terms of wall time. SODO has a better performance than outcome sampling MCCFR. Even though SODO here is using the exact best response, the graph still shows that SODO is a feasible and promising method to avoid full-tree traversal and minimize the estimated regret. Future works for scaling XODO can follow SODO but requires fast and accurate BR approximation because according to the regret bound of SODO, the error of BR approximation will limit the optimality of SODO.

\section{Conclusion}
This paper proposes the first generic framework for studying the theoretical convergence speed and algorithmic performance of regret-minimization based Double Oracle algorithms. Building upon this framework, we propose the Periodic Double Oracle algorithm which improves the sample complexity. Our numerical simulations demonstrate that PDO achieves superior performance compared to XDO and ODO in extensive-form games. Additionally, PDO exhibits fast convergence in games with small support NE, and remains robust across other games. Overall, our proposed framework and algorithm offer a significant contribution to the field of Double Oracle methods.

% A future direction can be an extension to the stochastic setting. Stochastic regret minimizer such as MCCFR can be used for restricted games solving. Approximate Best Response such as RL methods~\citep{perolat2022mastering, mcaleer2022escher} or IS-MCTS Best Response ~\citep{is_mcts_br} can be used. But due to the variance the stochastic methods might bring to RMDO, a smarter frequency function $m(j)$ is required. 
A future direction is combining PDO with deep regret-based methods for solving the restricted game~\citep{perolat2022mastering, mcaleer2022escher} and finding the BR~\citep{lanctot2017unified, mcaleer2022anytime, mcaleer2022self}.
As for applications of our framework, our framework can also be applied to robust and risk-aware reinforcement learning~\citep{zhang2020robust, lanier2022feasible, slumbers2022learning}. Additionally, it could be applied to solve extensive-form continuous games ~\citep{adam2021double} where $k$ can be large, but the sample complexity of PDO is only linear in $k$.

\section{Acknowledgements}
Yaodong Yang is supported in part by the National Key R\&D Program of China (20222D0114900). The authors would like to thank the Department of Statistical Science at University College London for providing the computing clusters. This work was supported by the Engineering and Physical Sciences Research Council [grant number EP/T517793/1, EP/W524335/1].


\bibliography{main}
\bibliographystyle{icml2023}

%%%%%%%%%%%%%%%%%%%%%%%%%%%%%%%%%%%%%%%%%%%%%%%%%%%%%%%%%%%%%%%%%%%%%%%%%%%%%%%
%%%%%%%%%%%%%%%%%%%%%%%%%%%%%%%%%%%%%%%%%%%%%%%%%%%%%%%%%%%%%%%%%%%%%%%%%%%%%%%
% APPENDIX
%%%%%%%%%%%%%%%%%%%%%%%%%%%%%%%%%%%%%%%%%%%%%%%%%%%%%%%%%%%%%%%%%%%%%%%%%%%%%%%
%%%%%%%%%%%%%%%%%%%%%%%%%%%%%%%%%%%%%%%%%%%%%%%%%%%%%%%%%%%%%%%%%%%%%%%%%%%%%%%
\newpage
\appendix
\onecolumn

\section{Additional Details of Extensive-form Double Oracle}
\label{append:xdo}
\subsection{Example of two iterations running}
\begin{figure}[h]
  \centering
  \includegraphics[width=.236\textwidth]{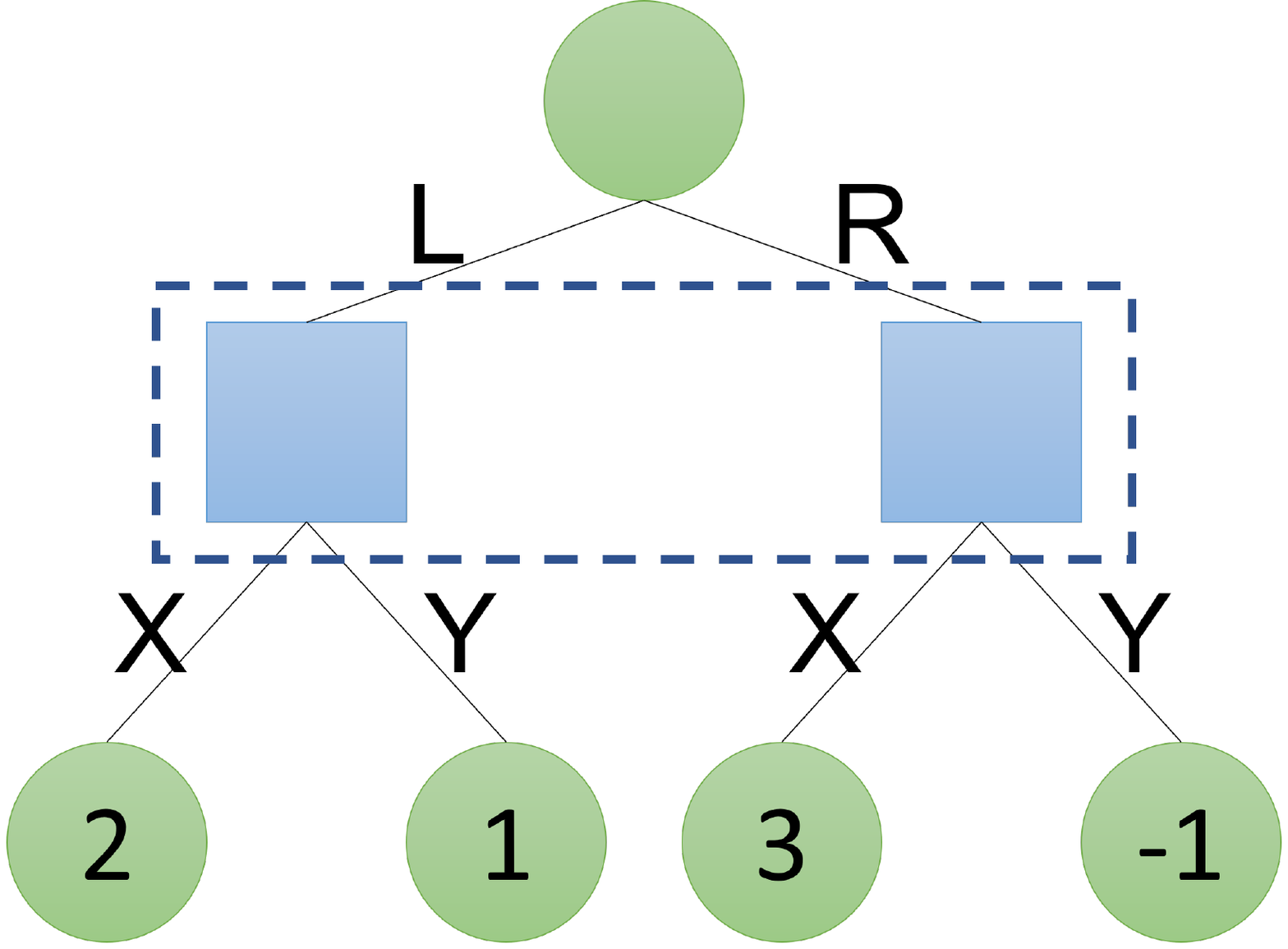}
  \includegraphics[width=.236\textwidth]{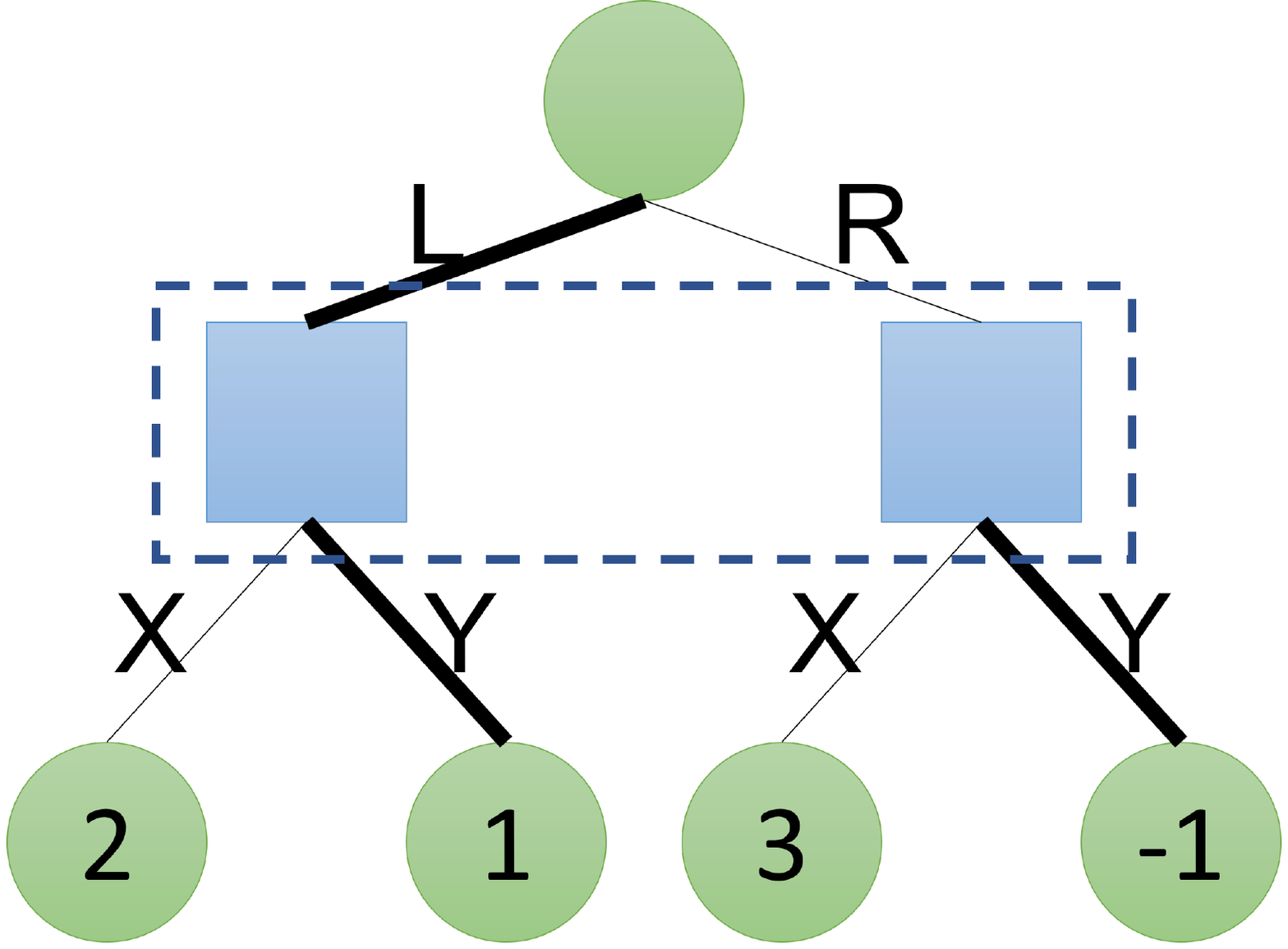}
\caption{An example of two iterations of Extensive-form Double Oracle. Square player is unaware of what Circle player chose. The values at the leafs are Circle player's payoff. Since it is zero-sum game, the Square player's payoffs are the opposite numbers.}
\label{fig:example}
\end{figure}
Here we present an illustration of two iterations of the Extensive-form Double Oracle in a zero-sum game in the figure above. In Iteration 0, both Circle and Square players have empty population and employ a uniform random strategy. In Iteration 1, Circle player's BR to Square's random strategy is action $L$ since $L$ and $R$'s expected values are $1.5$ and $1$, respectively. Square player's BR is action $Y$ since $X$ and $Y$'s expected utilities are $-2.5$ and $0$. Then $L$ and $Y$ are added to the population (thicker lines). Then since both players only have one action in the restricted game, the restricted game's NE is to choose them. Circle player's BR to NE of restricted game is still action $L$ given Square player will choose $Y$. Meanwhile, Square player's BR is still $Y$ given Circle player will choose $L$, since Square player's expected values of actions are $v(X)=-2$ and $v(Y)=-1$. As there is no BR actions added to the population, the restricted game won't change in latter iterations, DO converges and the NE for the original game is $(L,Y)$.

\subsection{Additional details of Algorithm \ref{alg:xdo}}
\label{append:algo_xdo}
Algorithm \ref{alg:xdo} has more detailed than the algorithm description of XDO in the original paper ~\citep{xdo}. In our description of XDO, we specify $\epsilon$ according to XDO's official implementation \href{https://github.com/indylab/tabular\_xdo/blob/main/main\_experiments.py\#L173}{https://github.com/indylab/tabular\_xdo/blob/main/main\_experiments.py\#L173}, where the authors start at a initial threshold $\epsilon=\epsilon_0$ and divide it by $2$ every time reaching $\epsilon$-NE in restricted game.

\subsection{Empirical values of $k$}
Table \ref{table:k} presents the findings of the empirical evaluation of the algorithmic execution of XDO across various games. It is observed that although the value of $k$ is usually smaller than the upper bound $|S|$, even for the smallest $k=16$, there is a considerable rise in both the number of iterations and the sample complexity caused by the exponential term: $4^{16} \sim \mathcal{O}(10^{9})$. The authors of XDO have proposed a solution to address this issue by setting a relatively large threshold value $\epsilon_0$. Nonetheless, the determination of the appropriate $\epsilon_0$ value remains a complex task since the sample complexity is highly sensitive to the value of $k$, which may vary based on the size of the NE support in different games.

\begin{table}[h]
\caption{The empirical values of $k$ when executing XDO to reach $10^{-3}$-NE and the number of infosets of different games}
\label{table:k}
\vskip 0.15in
\begin{center}
\begin{small}
\begin{sc}
\begin{tabular}{lcccr}
\toprule
Game Name & $k$ & Num. of Infosets \\
\midrule
Kuhn Poker          & 20& 58\\
Leduc Poker         & 16& 9457\\
Leduc Poker Dummy   & 17& 468517\\
Oshi Zumo           & 34& 60533\\
\bottomrule
\end{tabular}
\end{sc}
\end{small}
\end{center}
\vskip -0.1in
\end{table}

\section{Proofs}
\subsection{Proof of Lemma \ref{lemma:k}}
\label{append:k}

\begin{proof}
We first prove the upper bound. The number of time windows equals to the number of times adding new BR actions to the population. Then since the size of the population is less or equal to $\sum_{i}|S_i|$ and in each time window there is at least one new BR action added comparing to the last window, $k \leq \sum_{i}|S_i|$. 

Then we prove the lower bound of $k$. Since there is at most one pure strategy at a infoset added to the population every time when the restricted game is expanded and a new window starts. Thus at $s\in S$, $|\{a|a\in A(s)\}| \leq k$. Since at every infoset, the number of pure strategies in the converged population will be greater than the support of NE. Otherwise, there is a pure strategy in the NE strategy but not in the population, which means that the population doesn't converge and leads to contradiction. At $s\in S$, denote $\Pi^*$ as a set of NE of this game, $|\{a|a\in A(s)\}| \geq \min_{\pi\in \Pi^*}\text{supp}^{\pi}(s)$. So we have $k\geq \max_{s\in S}|\{a|a\in A(s)\}| \geq \max_{s\in S}\min_{\pi\in \Pi^*}\text{supp}^{\pi}(s)$.
\end{proof}

\subsection{Proof of Lemma \ref{lemma:bound_to_ne}}
\label{append:bound_to_ne}
\begin{proof}
Since the weighted-average strategy at time $T$ is defined as $\bar{\pi}_i=\sum_{t=1}^T \pi_i^t W_t$, where $\sum_t W_t=1$. Specifically, in Discounted CFR, $W_t=w_t/ \sum_t w_t$ and $\lim_{t\rightarrow \infty}w_t=\infty$. Since the value of strategy $v_i(\cdot, \cdot)$ is a linear function, then we have:
\begin{align}
    \max_{\pi^{'}_i} \sum_{t=1}^T [v_i(\pi^{'}_i, \pi_{-i}^t) - v_i(\pi_i^t, \pi_{-i}^t)] W_t 
    =
    \max_{\pi^{'}} v_i(\pi^{'}, \bar{\pi}_{-i}^t) - v_i(\bar{\pi_i}, \bar{\pi}_{-i}) 
    \leq 
    \mathcal{O}(|S_i|\sqrt{|A_i|}/\sqrt{T})\\
    \leq 
    \mathcal{O}(|S_i|\sqrt{|A|}/\sqrt{T})
\end{align}

Given the above inequality exists for both $i$, then in two-player zero-sum game setting, we have:
\[v_{-i}(\cdot,\cdot)=-v_{i}(\cdot,\cdot),\ \max_{\pi^{'}_{-i}} v_{-i}(\pi_i^{t}, \pi^{'}_{-i})=-\min_{\pi^{'}_{-i}} v_{i}(\bar{\pi}_i, \pi^{'}_{-i})\]
and then
\begin{equation}\label{eq: regret bound of OXDO 3}
    \begin{aligned}
    v_{i}(\bar{\pi})-\min_{\pi^{'}_{-i}} v_{i}(\bar{\pi}_i, \pi^{'}_{-i})\leq \mathcal{O}(|S_{-i}|\sqrt{|A|}/\sqrt{T}).
    \end{aligned}
\end{equation}
Therefore:
\begin{equation}\label{eq: regret bound of OXDO 5}
    \begin{aligned}
    \max_{\pi^{'}} v_i(\pi^{'}, \bar{\pi}_{-i}^t) - \sum_i \mathcal{O}(|S_i|\sqrt{|A|}/\sqrt{T})
    \leq 
    \max_{\pi^{'}} v_i(\pi^{'}, \bar{\pi}_{-i}^t) - \mathcal{O}(|S_i|\sqrt{|A|}/\sqrt{T}) \\
    \leq  
    v_{i}(\bar{\pi}) 
    \leq \min_{\pi^{'}_{-i}} v_{i}(\bar{\pi}_i, \pi^{'}_{-i}) + \mathcal{O}(|S_{-i}|\sqrt{|A|}/\sqrt{T})
    \leq 
    \min_{\pi^{'}_{-i}} v_{i}(\bar{\pi}_i, \pi^{'}_{-i}) + \sum_i \mathcal{O}(|S_{-i}|\sqrt{|A|}/\sqrt{T}),
    \end{aligned}
\end{equation}
$\bar{\pi}$ is $\epsilon_T$-NE, where $\epsilon=\sum_i \mathcal{O}(|S_i|\sqrt{|A|}/\sqrt{T}) = \mathcal{O}(|S|\sqrt{|A|}/\sqrt{T})$.
\end{proof}

\subsection{Proof of Theorem \ref{theo:RMDO_bound}}
\label{append:RMDO_bound}
\begin{proof} Since iteration $0$ population directly add BR with respect to the uniform random strategy, $|A_{i,0}|=1$, where  $|A_{i,j}|$ denotes player $i$'s maximal number of actions over all infoset in time window $T_j$. Since we have $k$ time windows in the period $T$, we then have $\sum_{j=0}^{k-1} |T_j| = T$. $w_t$ is the within window weights. In this time window $T_j$, at most $1$ pure action is added to the action set in an infoset compared to the previous time window, then the maximum number of actions at infosets of time window $T_j$ will be bounded by:
\[|A_{i,j}|\leq j \leq k.\]

Denote $V_j$ as the set of iterations after last BR checking point in $T_j$. In the last window $T_{k-1}$, according to assumption \ref{assume:k}, we have $V_j = \varnothing$ since $\lim_{T\rightarrow \infty}|T_{k-1}|=\infty$, there is no "last" since the final window will not end. 

Note that since the population will stop growing  Every time after computing the BR, RMDO will check if the BR to the average strategies of the opponent is in the current population. Therefore at two different iterations in the same time window, since the populations, we have the local Best Response is exactly the global Best response:
\begin{align}
    \max_{\pi'_i \in \Pi_j} \sum_{t\in T_j\setminus V_j} w_t \cdot v_i(\pi'_i, \pi_{-i}^{t}) = \max_{\pi'_i \in \Pi}\sum_{t\in T_j\setminus V_j} w_t \cdot v_i(\pi'_i, \pi_{-i}^{t}).
\label{eq:tw}
\end{align}

Since the frequency of BR computing in current window $T_j$ is $m(j)$, then $|V_j|<m(j)$ for $j\leq k-1$. Given the regret minimizer has $\mathcal{O}(|S_i|\sqrt{|A_i|T^{-1}})$ (weighted) average regret, then in window $T_j$ for $j=0, 1, \cdots, k$, we have:
\begin{align}\label{eq:ox1}
    \max_{\pi'_i \in \Pi_j} \sum_{t\in T_j\setminus V_j} w_t \cdot v_i(\pi'_i, \pi_{-i}^{t}) - \sum_{t\in T_j\setminus V_j} w_t \cdot v_i(\pi_i^{t}, \pi_{-i}^{t})
    \leq
    \mathcal{O}(|S_i|\sqrt{|A^{T_j}_i| \cdot|T_j\setminus V_j|^{-1}} 
    \leq
    \mathcal{O}(|S_i|\sqrt{k\cdot |T_j\setminus V_j|^{-1}}).
\end{align}
Thus:
\begin{align}
    \max_{\pi'_i \in \Pi_j} \sum_{t\in T_j\setminus V_j} W_t \cdot v_i(\pi'_i, \pi_{-i}^{t}) - \sum_{t\in T_j\setminus V_j} W_t \cdot v_i(\pi_i^{t}, \pi_{-i}^{t})
    \leq
    \frac{|T_j|}{T} \cdot \mathcal{O}(|S_i|\sqrt{k\cdot |T_j\setminus V_j|^{-1}})
\end{align}

 Then the weighted-average regret:
\begin{align}
    &R_i^T = \max_{\pi^{'}_i} \sum_{t=1}^T W_t \cdot (v_i(\pi^{'}_i, \pi_{-i}^t) - v_i(\pi^t)) 
    = \sum_j\left[ \max_{\pi'_i \in \Pi} \sum_{t\in T_j} W_t \cdot v_i(\pi'_i, \pi_{-i}^{t}) - \sum_{t\in T_j} W_t \cdot v_i(\pi_i^{t}, \pi_{-i}^{t})\right] \\ 
    &= \sum_j \left[
    \max_{\pi'_i \in \Pi} \sum_{t\in T_j\setminus V_j} W_t \cdot v_i(\pi'_i, \pi_{-i}^{t}) - \sum_{t\in T_j\setminus V_j} W_t \cdot v_i(\pi_i^{t}, \pi_{-i}^{t}) + \max_{\pi'_i \in \Pi} \sum_{t\in V_j} W_t \cdot v_i(\pi'_i, \pi_{-i}^{t}) - \sum_{t\in V_j} W_t \cdot v_i(\pi_i^{t}, \pi_{-i}^{t}) \right] \\
    &= \sum_j\left[
    \max_{\pi'_i \in \Pi_j} \sum_{t\in T_j\setminus V_j} W_t \cdot v_i(\pi'_i, \pi_{-i}^{t}) - \sum_{t\in T_j\setminus V_j} W_t \cdot v_i(\pi_i^{t}, \pi_{-i}^{t}) +
    \max_{\pi'_i \in \Pi} \sum_{t\in V_j} W_t \cdot v_i(\pi'_i, \pi_{-i}^{t}) - \sum_{t\in V_j} W_t \cdot v_i(\pi_i^{t}, \pi_{-i}^{t}) \right] \\
    &\leq \sum_j \frac{|T_j|}{T} \cdot\mathcal{O}\left(
    |S_i|\sqrt{k} \cdot |T_j\setminus V_j|^{-1/2} + \sum_{t\in V_j} w_t
    \right) 
    < 
    \sum_j \frac{|T_j|}{T} \cdot\mathcal{O}\left(
    |S_i|\sqrt{k} \cdot |T_j\setminus V_j|^{-1/2} + \sum_{t\in V_j} 1
    \right)\\ 
    &\leq
    \frac{1}{T} \cdot\mathcal{O}\left( |S_i|\sqrt{k} \cdot \sum_{j=0}^{k-1}\frac{|T_j|}{\sqrt{|T_j| - m(j) + 1}} + \sum_{j=0}^{k-2}|T_j| [m(j)-1]
    \right).
    \end{align}

The last but one inequality exists becuase $|T_j\setminus V_j|^{-1/2} \leq (|T_j| - m(j) + 1)^{-1/2}$, $V_{k-1}=\varnothing$ and $|V_{j}|\leq m(j)-1$.

Therefore, the upper bound of weighted-average regret:
\begin{align}
    \bar{R}_i^T = \max_{\pi^{'}_i} \sum_{t=1}^T W_t \cdot (v_i(\pi^{'}_i, \pi_{-i}^t) - v_i(\pi^t))
    < 
    \mathcal{O}\left( \sum_{j=0}^{k-2} \frac{|T_j| [m(j)-1]}{T} + \frac{|S_i|\sqrt{k}}{T} \cdot \sum_{j=0}^{k-1}\frac{|T_j|}{\sqrt{|T_j| - m(j) + 1}}
    \right).
\end{align}

According to Assumption \ref{assume:k}, $\lim_{T\rightarrow \infty} |T_{k-1}|/T=1$ and $\lim_{T\rightarrow \infty} \sum_{j=0}^{k-2}|T_j|/T=0$. Then we have:
\begin{align}
    \lim_{T\rightarrow \infty} m(k-1)/|T_{k-1}| =  \lim_{T\rightarrow \infty} m(k-1)/ T \cdot \lim_{T\rightarrow \infty} T / |T_{k-1}| = 0 \cdot 1 = 0
\end{align}

Then if $m(j)$ is sublinear, meaning that $\forall j,\ \lim_{T\rightarrow \infty} m(j)/T=0$, then the weighted-average regret satisfies that:
\begin{align}
    &\lim_{T\rightarrow \infty}\frac{1}{T} \cdot\mathcal{O}\left( \sum_{j=0}^{k-2}|T_j| [m(j)-1] + |S_i|\sqrt{k} \cdot \sum_{j=0}^{k-1}\frac{|T_j|}{\sqrt{|T_j| - m(j) + 1}}
    \right) \\
    &\leq
    \lim_{T\rightarrow \infty}\frac{1}{T} \cdot \mathcal{O}(\sum_{j=0}^{k-2}|T_j| [m(j)-1]) + 
   \lim_{T\rightarrow \infty}\frac{1}{T} \cdot \mathcal{O}(\frac{|S_i|\sqrt{k}\sqrt{|T_{k-1}|}}{T\sqrt{1 - m(k-1)/\sqrt{T_{k-1}}) + 1/\sqrt{T_{k-1}}}} + \\
   &\lim_{T\rightarrow \infty}\frac{1}{T} \cdot \mathcal{O}(|S_i|\sqrt{k} \cdot \sum_{j=0}^{k-2}\sqrt{|T_j|})
    =
    \lim_{T\rightarrow \infty}\frac{1}{T} \cdot \mathcal{O}(|S_i|\sqrt{k}\sqrt{|T_{k-1}|}) = 0.
\end{align}
\end{proof}

\subsection{Proof of Theorem \ref{theo:lw-xodo}}
\label{append:lw-xodo_it}
\begin{proof}
Before reaching the global $\epsilon$-NE, in each windows $T_j,\ j<k-1$, the number of iterations must satisfy that $|T_j|<\mathcal{O}(|A||S|^2/\epsilon^2 + m(j)-1)$ if the regret minizer has $\mathcal{O}(\sqrt{|A|}|S|/T)$ regret. Since if $|T_j|\geq\mathcal{O}(|A||S|^2/\epsilon^2 + m(j)-1)$, there must be a BR computing happening after $\mathcal{O}(|A||S|^2/\epsilon^2)$ iterations, and at this iteration $\mathbf{BR}(\pi^t)\in \Pi_j$ and $\pi^t$ has reached restricted game's local $\epsilon$-NE, then $\pi^t$ is the $\epsilon$-NE in the original game. Besides, in $T_{k-1}$, we also need less than $\mathcal{O}(|A||S|^2/\epsilon^2 + m(j)-1)$ iterations to reach $\epsilon$-NE. Then sum up iterations in all time window, we get that the expected number of iterations is less than $\mathcal{O}(k|A||S|^2/\epsilon^2 - k + \sum_{j=0}^{k-1}m(j))$.
\end{proof}

\subsection{Proof of Theorem \ref{theo:xodo_sc}}
\label{append:xodo_sc}
\begin{proof}
When finding $\epsilon$-NE, XODO needs $\mathcal{O}(|S^2k^2/\epsilon^2)$ number of iterations and thus compute $\mathcal{O}(|S^2k^2/\epsilon^2)$ times of BR. Given in the worst case, the complexity of traverse the entire game tree once and computing BR is both $|S|$, the overall sample complexity is $\mathcal{O}(|S^2k^2/\epsilon^2) \cdot |S| + \mathcal{O}(|S^2k^2/\epsilon^2) \cdot |S| = \mathcal{O}(2|S^2k^2/\epsilon^2)$.
\end{proof}

\subsection{Proof of Corollary \ref{theo:xdo_cr}}
\label{append:gxdo_sc}
In XDO, when the stopping threshold of restricted games is divided by $\alpha>1$ at the end of each window, XDO needs the following number of iterations to reach $\epsilon$-NE $\mathcal{O}(k|A||S|^2/\epsilon^2 + |A||S|^2 \alpha^{2k}/\epsilon_0^2 - k)$ ($\alpha=2$ in the original XDO).
\begin{proof}
Given in each time window $T_j$, $m(j) \geq |A||S|^2 \alpha^{2k}/\epsilon_0^2$, plug it into Theorem \ref{theo:lw-xodo}, we have the required number of iterations to reach $\epsilon$-NE is $\mathcal{O}(k|A||S|^2/\epsilon^2 - k + |A||S|^2 (\alpha^{2k} - 1)/[(\alpha^2 - 1)\epsilon_0^2]) \sim \mathcal{O}(k|A||S|^2/\epsilon^2 + |A||S|^2 \alpha^{2k}/\epsilon_0^2])$.
\end{proof}

\subsection{Proof of Theorem \ref{theo:xdo_sc}}
\label{append:xdo_sc}
\begin{proof}
When finding $\epsilon$-NE, XDO needs $\mathcal{O}(k|A||S|^2/\epsilon^2 + |S| 4^k/\epsilon_0^2 - k)$ number of iterations. Before convergence, XDO expands the restricted game $k$ times and manage to expand every time when it computes BR. Similarly, given in the worst case, the complexity of traverse the entire game tree once and computing BR is both $|S|$, the overall sample complexity is $\mathcal{O}(k|A||S|^2/\epsilon^2 + |A||S|^2 4^k/\epsilon_0^2 - k + k) \cdot |S| = \mathcal{O}(k|A||S|^3/\epsilon^2 + |A||S|^3 4^k/\epsilon_0^2)$.
\end{proof}

\subsection{Proof of Theorem \ref{theo:pdo_sc}}
\label{append:pdo_sc}
\begin{proof}
When finding $\epsilon$-NE, PDO needs $\mathcal{O}(k|A||S|^2/\epsilon^2 + (c-1)k)$ number of iterations and thus compute $\mathcal{O}(k|A||S|^2/c\epsilon^2 + (c-1)k/c)$ times of BR. Given in the worst case, the complexity of traverse the entire game tree once and computing BR is both $|S|$, the overall sample complexity is 
\begin{align}
   \mathcal{O}(k|A||S|^2/\epsilon^2 + (c-1)k + k|A||S|^2/c\epsilon^2 + (c-1)k/c) \cdot |S| \\
   = \mathcal{O}(k|A||S|^3/\epsilon^2 + ck|S| + k|A||S|^3/c\epsilon^2 - k|S|/c).
\end{align}

\end{proof}

\section{Description of the Games}
\label{append:games}
\textit{Sequential Blotto} is a revised sequential version of discrete Colonel Blotto Game. Each player at the beginning has a set of forces with different strength and need to put them on board in turns for fighting. Specifically, when the parameters is set to $(4,25)$, there will be $25$ forces for each person with strength from $0$ to $24$ and $4$ times of putting forces on board, thus $2$ rounds in total. In each round, players choose a force in turns. The battle results equal to the difference between the power of two forces on board. At the end of this round, the forces will be removed from board and start the next round. The payoff will only appear at the end of the game, equal to the summation of all battle results.

\textit{Large Kuhn Poker} is a variant of Kuhn Poker where a initial pot for each player is $40$. Players can bet any remaining amount.

\textit{Leduc Poker Dummy} is the same as vanilla Leduc Poekr except the actions in each information set are duplicated twice~\citep{xdo}.

\textit{Oshi Zumo} A board game in which players must repeatedly bid to push a token off the other side of the board. ~\citep{oshizumo}. In the instance of our experiment, there are two players each of whom is initialized with $4$ coins and the token is originally at the center of a board with length $2K+1$. $K$ in our case is $6$. Players need to put their bid on the board from the amount of coins they have, which is at least $M$. $M$ in our setting is $1$. Then the player who have chosen the greater number can push the token one step toward its opponent. Then both player needs to remove the number of coins in their bids. The winner is the player successfully push the token off the side of the board of its opponent. Winner will get payoff $1$ and the other will get $-1$.

We then offer some rough analysis on the support information in limited number common games. There are many more common games shown with small Nash support introduced in Table 2 of ODO paper ~\citep{odo}.

\begin{table*}[h]
\label{table:support}
\vskip 0.15in
\begin{center}
\begin{small}
\begin{sc}
\begin{tabular}{lcccr}
\toprule
Games & Minimum support percentage & Maximum support percentage\\
\midrule
Sequential Blotto (n, m)               & $1/M$ & $100\%$ \\
Kuhn Poker                 & $50\%$ & $100\% $\\
Leduc Poker Dummy & $\geq 25\%$ & $\leq 50\%$ \\
\bottomrule
\end{tabular}
\end{sc}
\end{small}
\end{center}
\vskip -0.1in
\caption{Percentage of Nash support in benchmark games, which is the support size of NE divided by the number of available actions at the corresponding infoset. Since at different infoset we have different support size of NE and number of available actions, we offer minimum and maximum support percentage. The minimum support of sequential Blotto's NE will be the similar idea of tic for tac. Given it is perfect information in our setting, the min player will have a pure strategy NE. The NE of Kuhn Poker is computed in ~\citep{wiki:Kuhn_poker}. In Leduc Poker Dummy, since actions a duplicated once, the maximum support will only half of number of available actions.}
\end{table*}

\section{Additional details of Periodic Double Oracle}
\label{append:pdo_add}
\begin{algorithm}[h]
\begin{algorithmic}
    \caption{Periodic Double Oracle}
    \label{alg:pdo}
    \STATE Hyperparameter $m(j)=c$, initial population $\Pi_0$, window index $j=0$.
    \FOR{$t=0,\cdots,\infty$}
    \STATE Construct restricted game $\mathbf{G}_t$ with $\Pi_t$.
    \STATE Update $\pi^t$ in $\mathbf{G}_t$.
    \IF{$t \mod c=0$}
    \STATE Compute $\Tilde{\pi}_i^t$ with equation (\ref{eq:avg_strategy}).
    \FOR{$i\in \{1,2\}$}
    \STATE
    $\Pi_{t+1} = \Pi_t \cup  \mathbf{BR}_i(\Tilde{\pi}^t_{-i})$.
    \ENDFOR
    \IF{$\Pi_t \neq \Pi_{t-1}$}
    \STATE Start new window: $j = j+1$
    \STATE Reset strategy $\pi^{t+1}$.
    \ENDIF
    \STATE \textbf{Output} $\bar{\pi}^t$.
    \ENDIF
    \ENDFOR
\end{algorithmic}
\end{algorithm}

\begin{table*}[h]
\caption{The sample complexity of examples of RMDO to reach $\epsilon$-NE in an extensive-form game where $|S|$ is the number of infosets and $|A|=\max_{s\in S} |A(s)|$.}
\label{table:sc}
\vskip 0.15in
\begin{center}
\begin{small}
\begin{sc}
\begin{tabular}{lcccr}
\toprule
Example & Sample Complexity & Sample Complexity in $k$\\
\midrule
XODO                & $\mathcal{O}(2|S|^3k^2/\epsilon^2)$ & polynomial\\
XDO                 & $\mathcal{O}(k|A||S|^3/\epsilon^2 + |A||S|^3 4^k/\epsilon_0^2)$ & exponential\\
PDO         & $\mathcal{O}(k|A||S|^3/\epsilon^2 + ck|S| + k|A||S|^3/c\epsilon^2 - k|S|/c)$ & linear\\
\bottomrule
\end{tabular}
\end{sc}
\end{small}
\end{center}
\vskip -0.1in
\end{table*}

\section{Additional Experimental Results}
\label{append:graphs}
\vspace{1.7cm}

\begin{figure*}[h]
     \centering
\begin{subfigure}
    \centering
    \includegraphics[width=.49\textwidth]{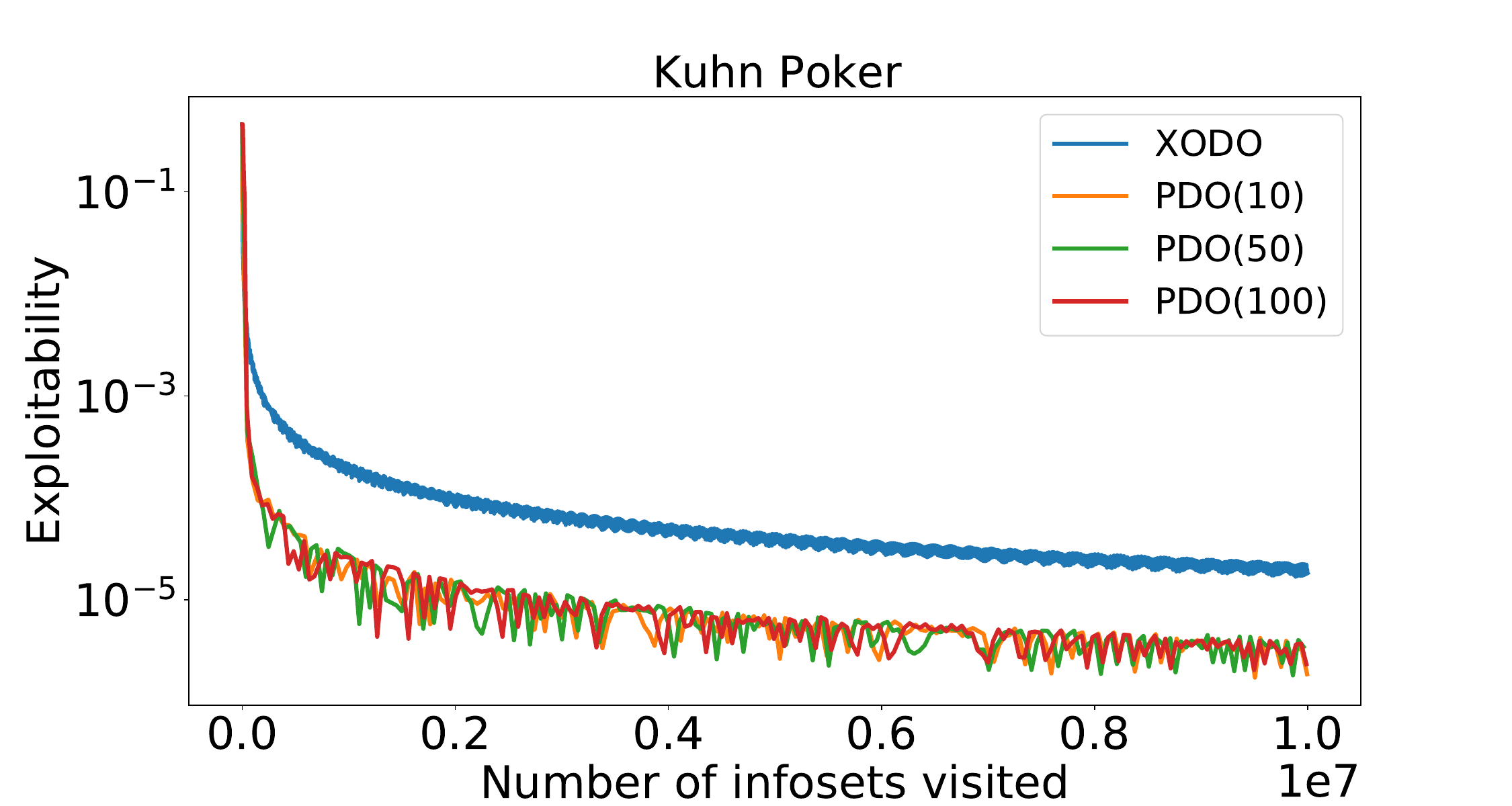}
     \label{fig:oshi_info}
\end{subfigure}
\begin{subfigure}
    \centering
    \includegraphics[width=.49\textwidth]{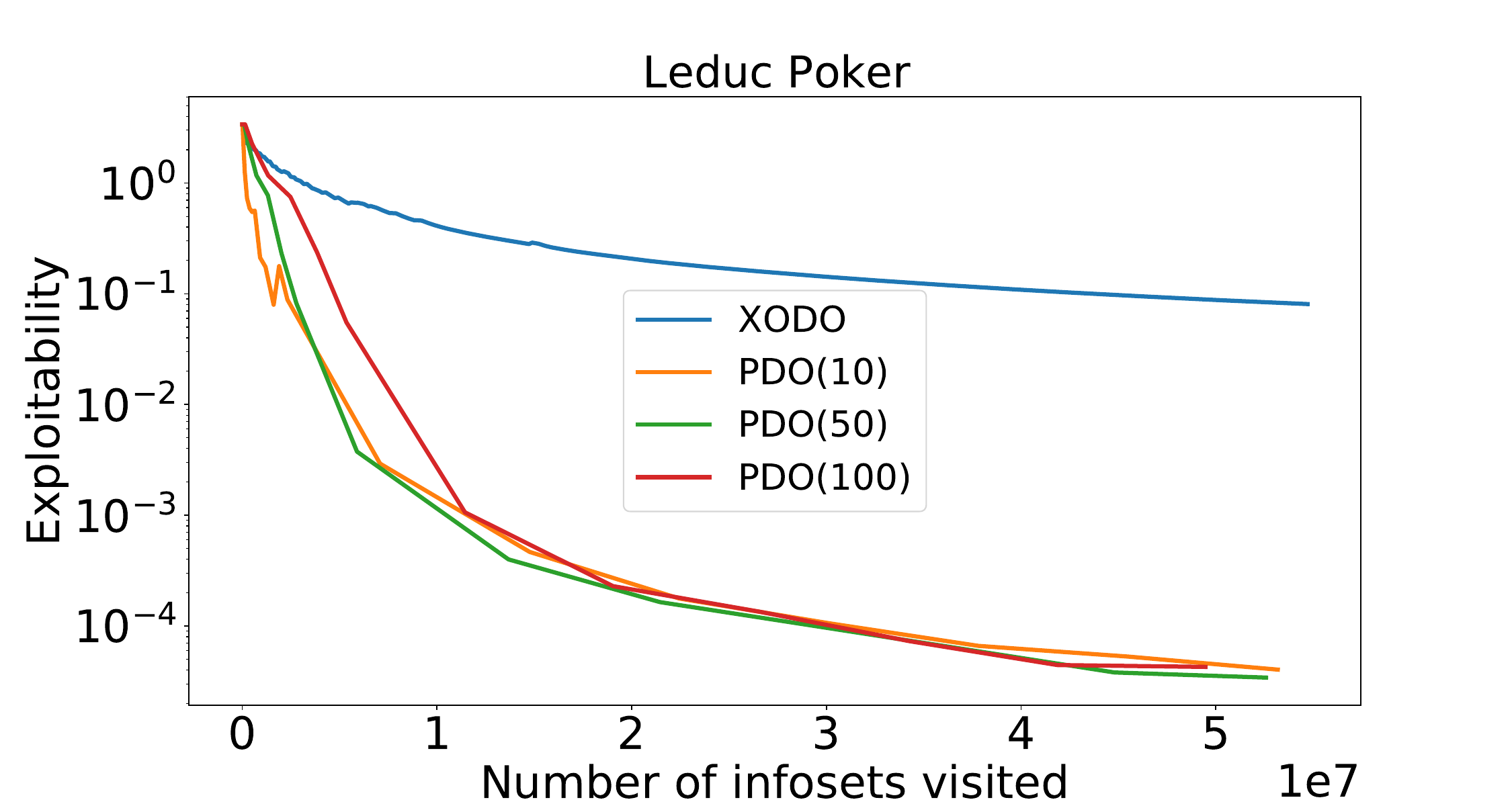}     
    \label{fig:dl_info}
\end{subfigure}
% \caption{Performance on Oshi Zumo}
%   \label{fig:exp_oshi}
% \end{figure*}

% \begin{figure*}[h]
%      \centering
\begin{subfigure}
    \centering
    \includegraphics[width=.49\textwidth]{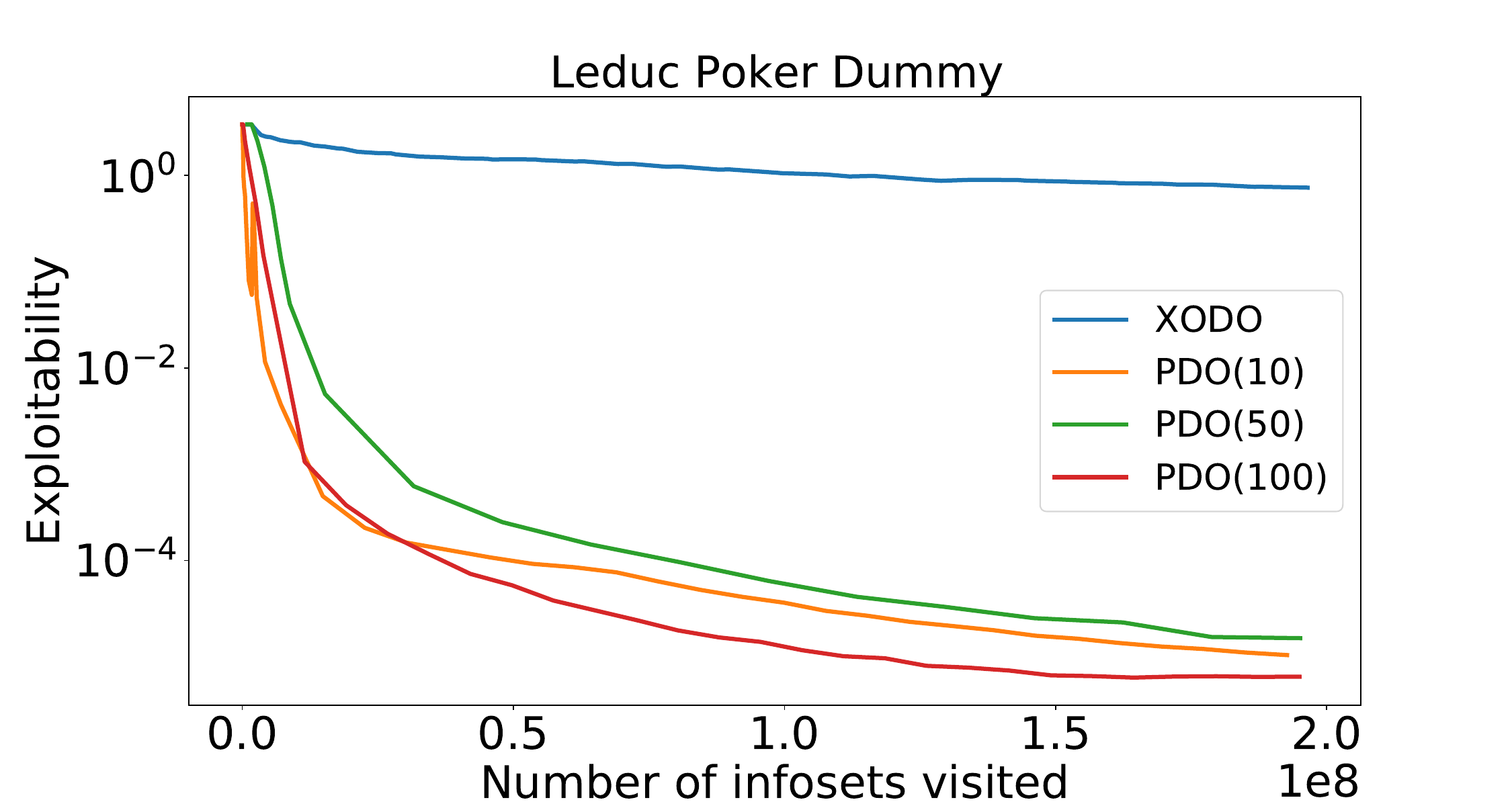}
     \label{fig:oshi_info}
\end{subfigure}
\begin{subfigure}
    \centering
    \includegraphics[width=.49\textwidth]{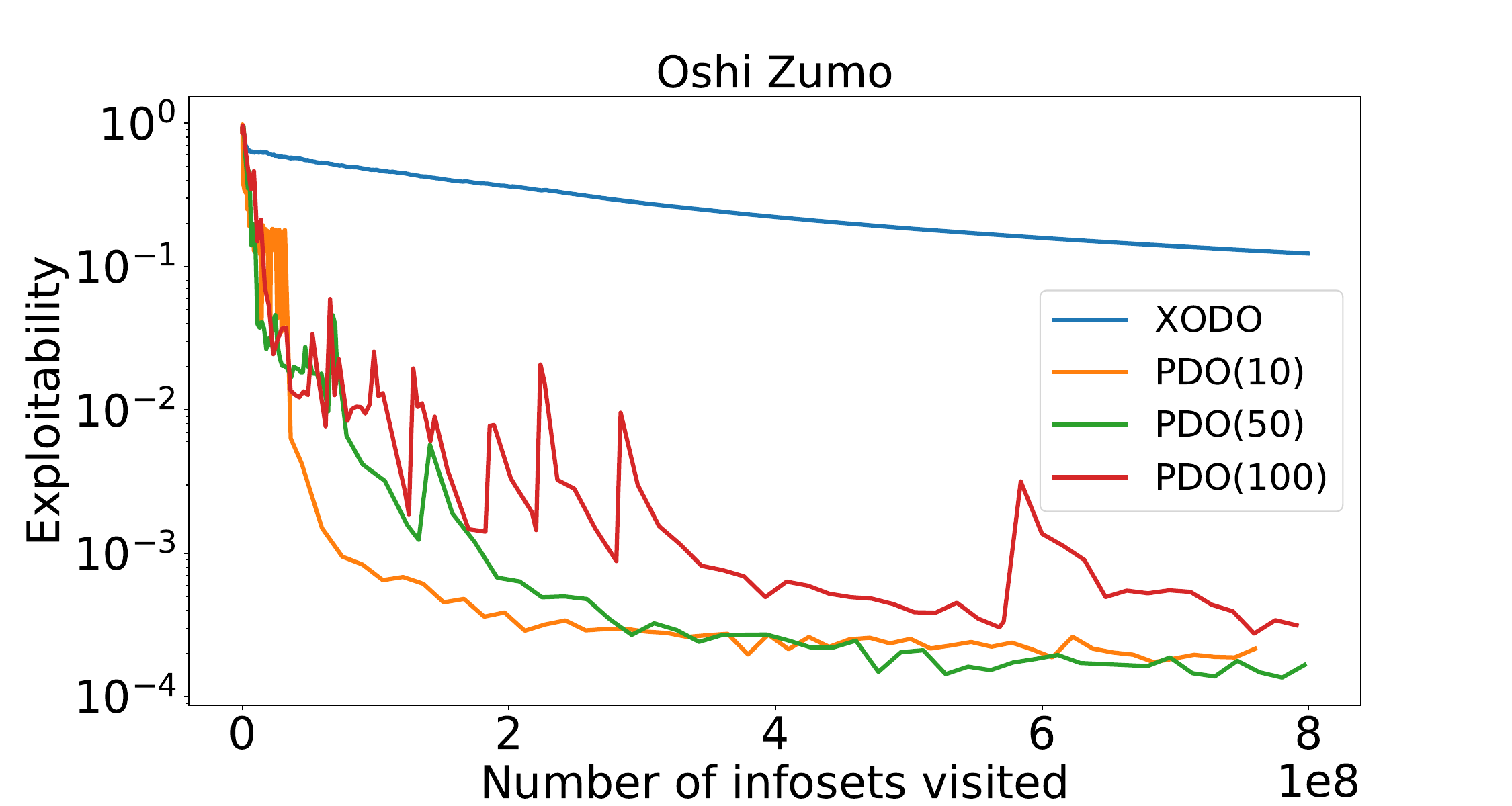}
    \label{fig:dl_info}
\end{subfigure}
\caption{Performance of PDO with different periodicity and XODO in terms of number of infosets. PDO significantly outperforms XODO.}
  \label{fig:exp_dlp}
\end{figure*}

\begin{figure*}[h]
     \centering
\begin{subfigure}
    \centering
    \includegraphics[width=.49\textwidth]{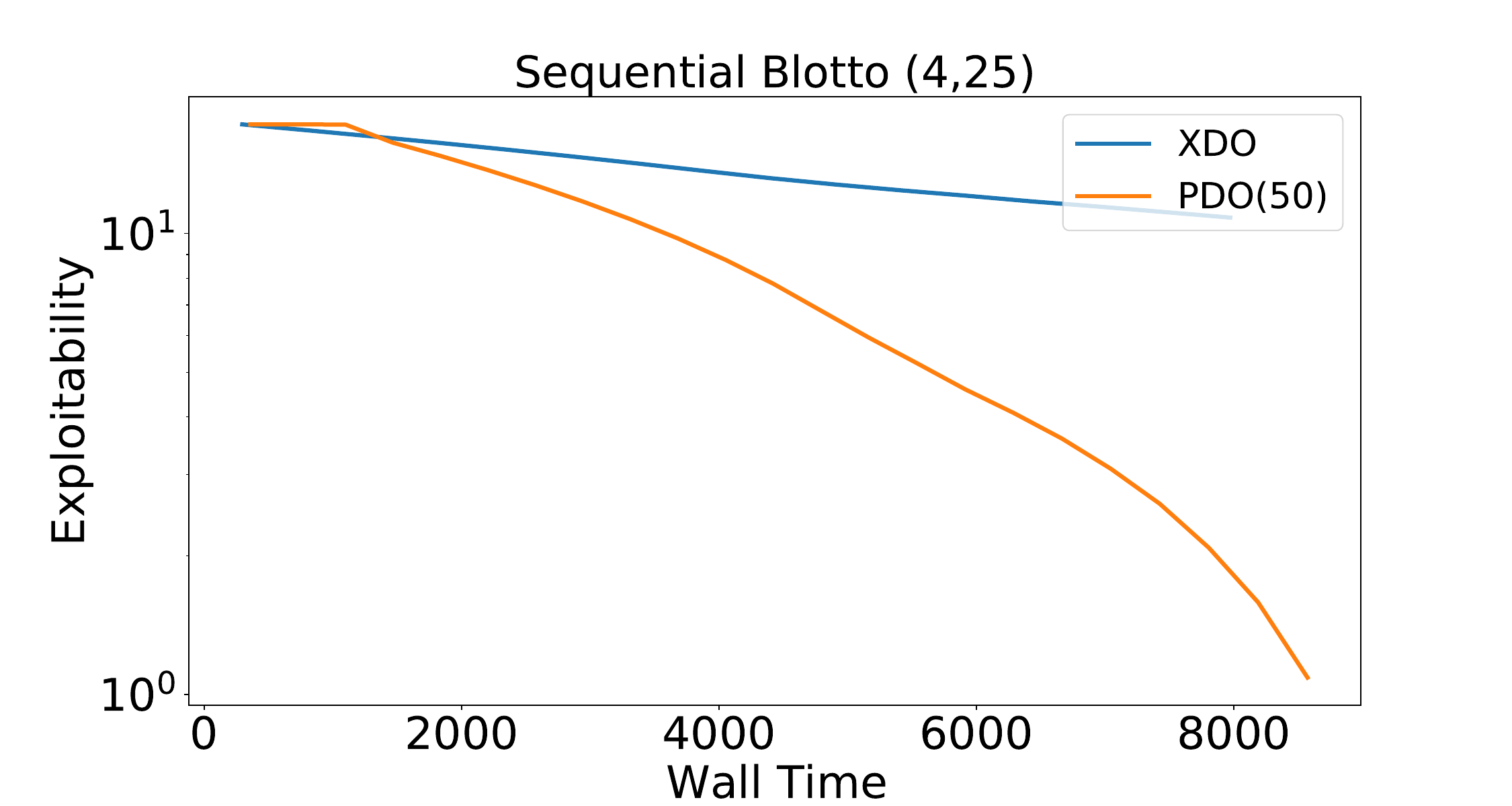}
     \label{fig:oshi_info}
\end{subfigure}
\begin{subfigure}
    \centering
    \includegraphics[width=.49\textwidth]{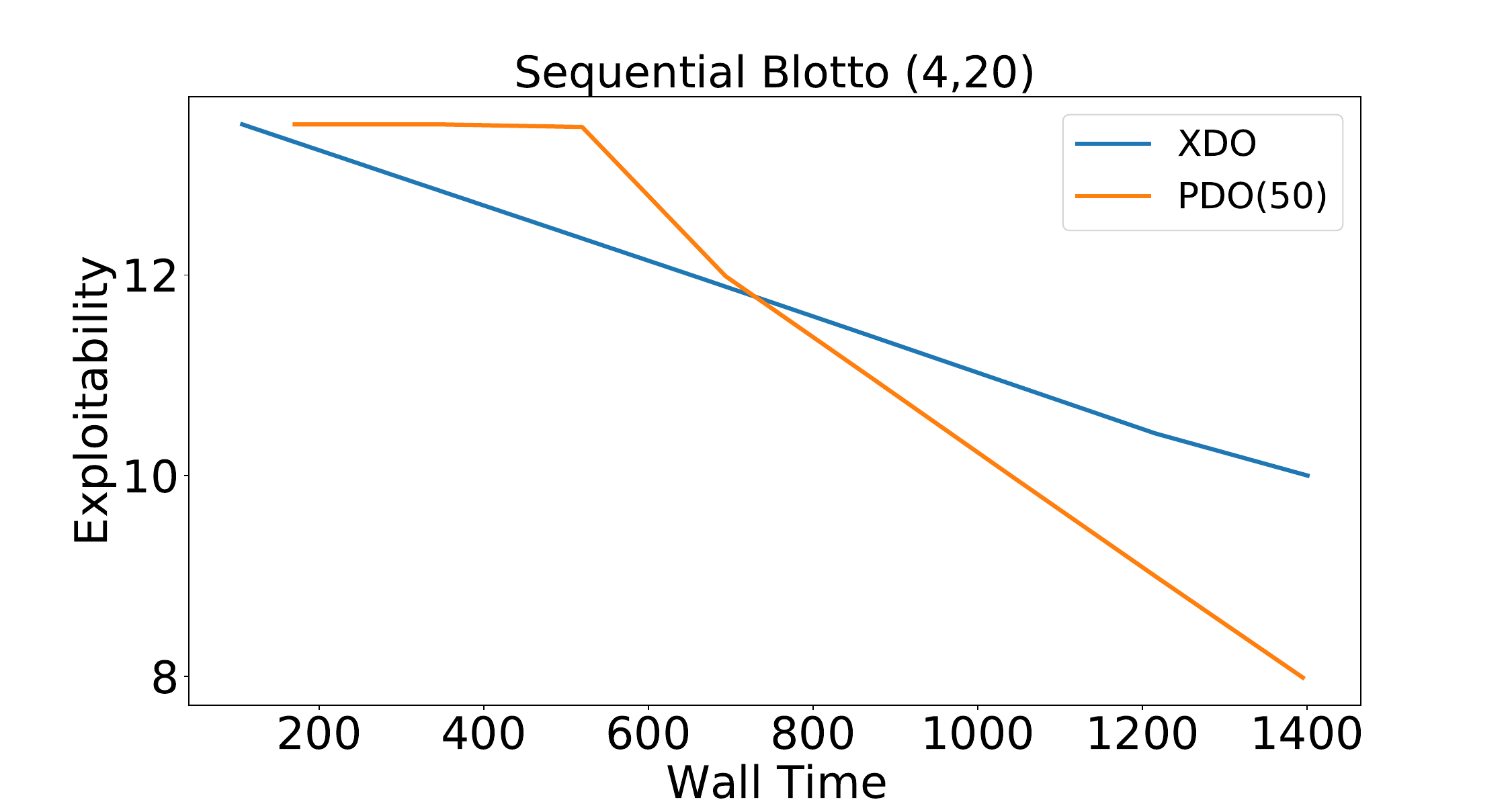}
    \label{fig:dl_info}
\end{subfigure}
\caption{Performance on Sequential Blotto in terms of wall time (in seconds).}
  \label{fig:exp_dlp}
\end{figure*}

\begin{figure*}[h]
     \centering
\begin{subfigure}
    \centering
    \includegraphics[width=.49\textwidth]{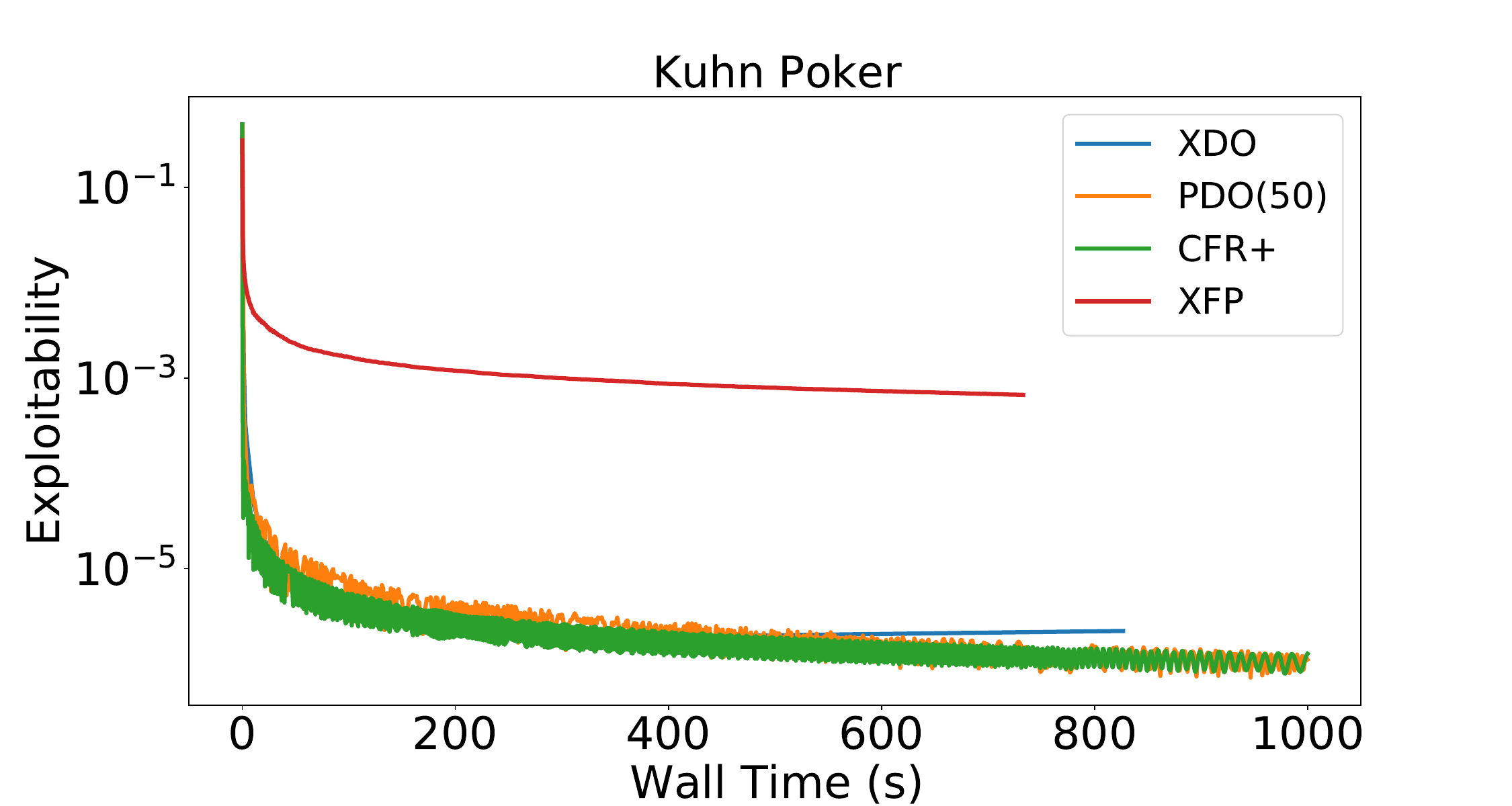}
     \label{fig:oshi_info}
\end{subfigure}
\begin{subfigure}
    \centering
    \includegraphics[width=.49\textwidth]{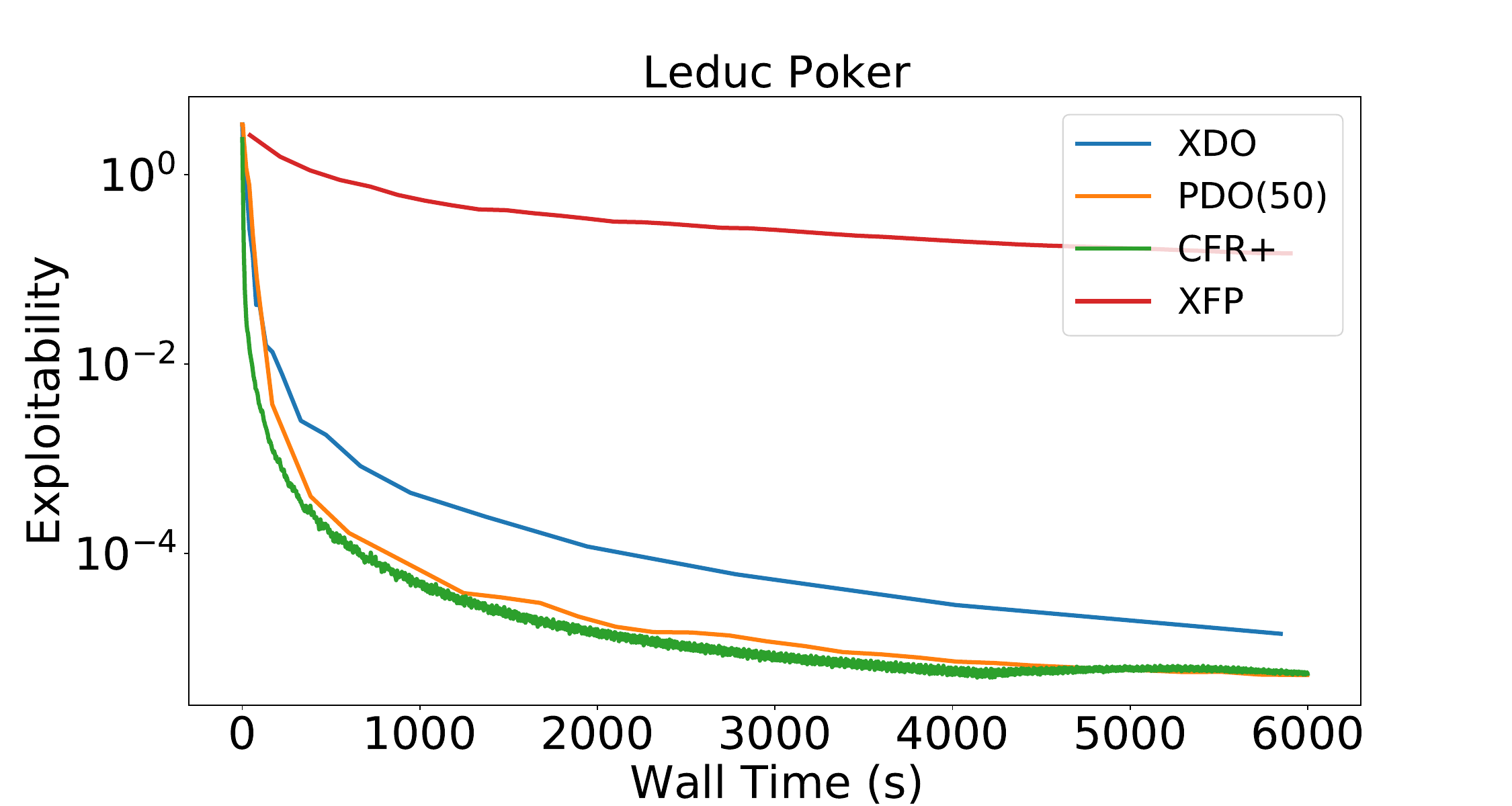}     
    \label{fig:dl_info}
\end{subfigure}
% \caption{Performance on Oshi Zumo}
%   \label{fig:exp_oshi}
% \end{figure*}

% \begin{figure*}[h]
%      \centering
\begin{subfigure}
    \centering
    \includegraphics[width=.49\textwidth]{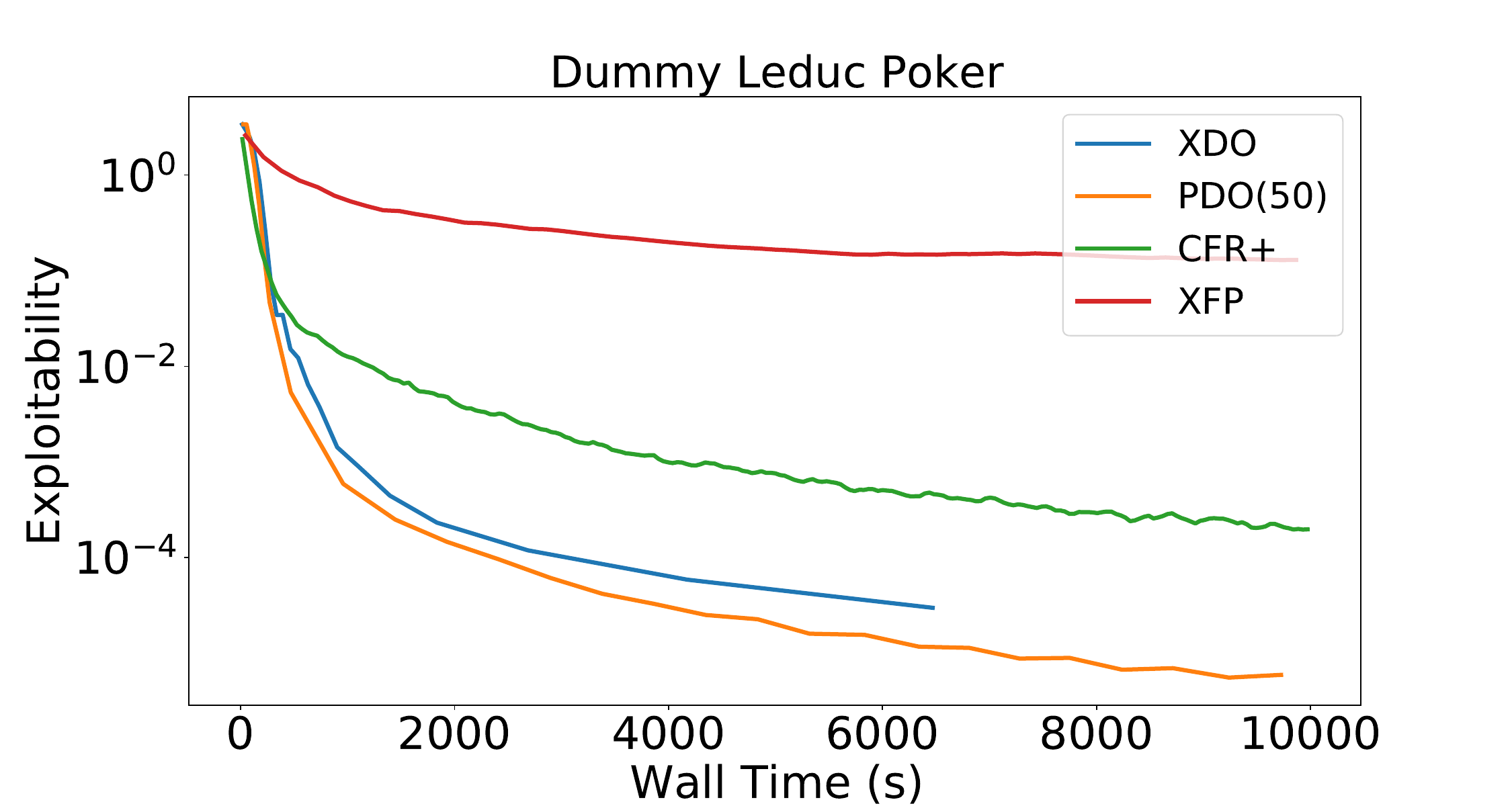}
     \label{fig:oshi_info}
\end{subfigure}
\begin{subfigure}
    \centering
    \includegraphics[width=.49\textwidth]{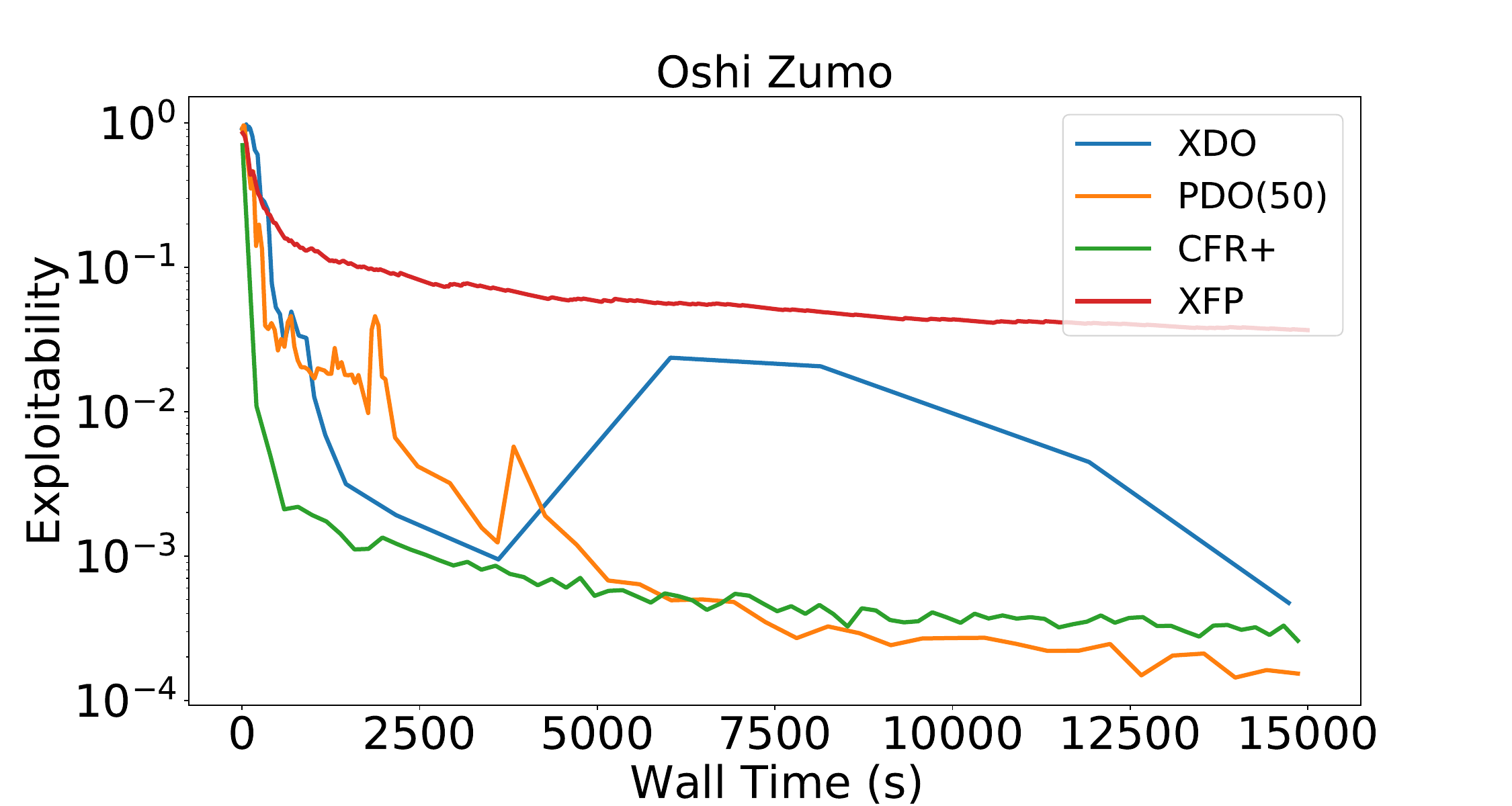}
    \label{fig:dl_info}
\end{subfigure}
\caption{Performance of PDO ($50$) comparing to baselines.}
  \label{fig:exp_dlp}
\end{figure*}

\begin{figure*}[t!]
     \centering
\begin{subfigure}
    \centering
    \includegraphics[width=.32\textwidth]{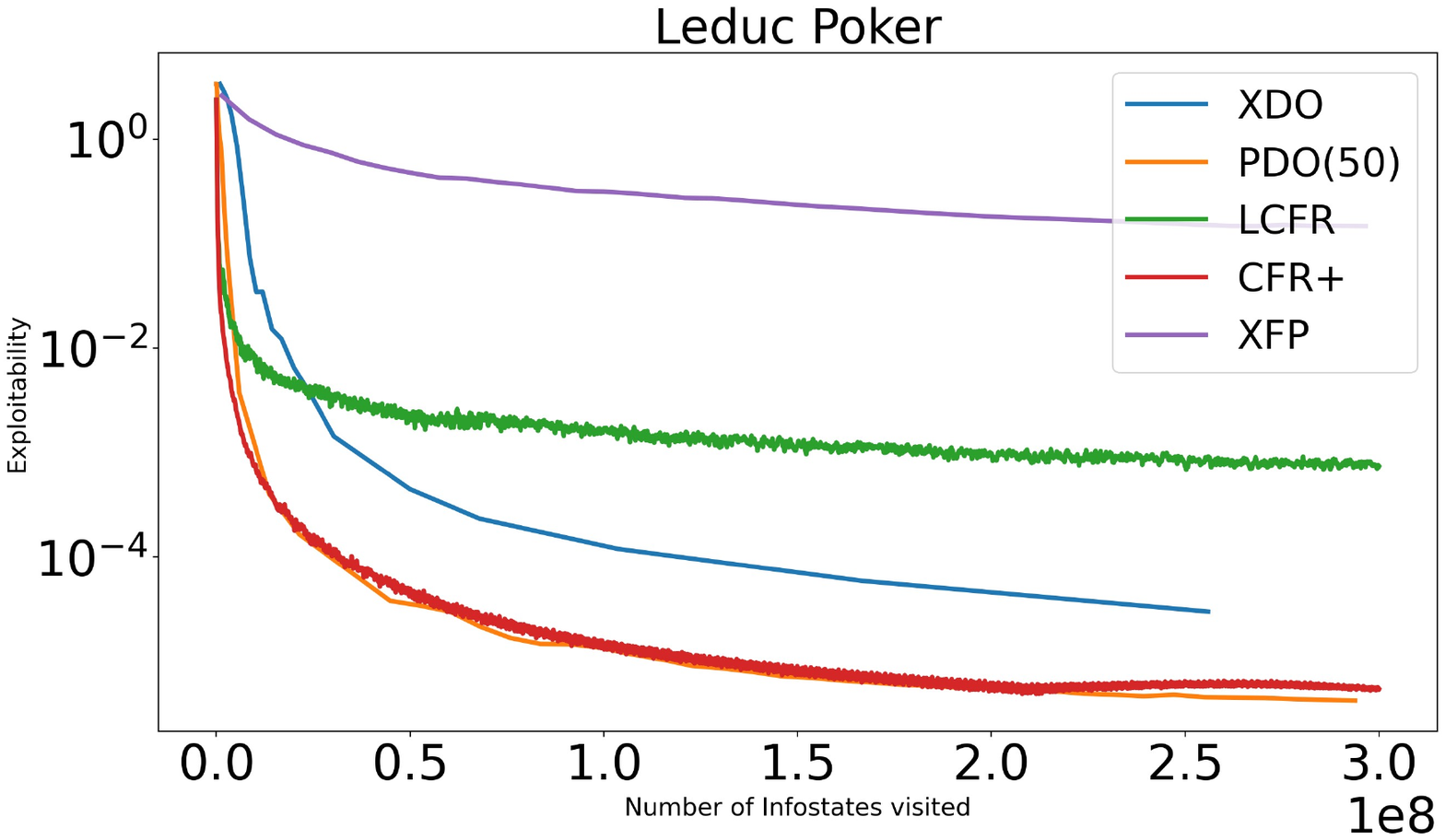}     
\end{subfigure}
\begin{subfigure}
    \centering
    \includegraphics[width=.32\textwidth]{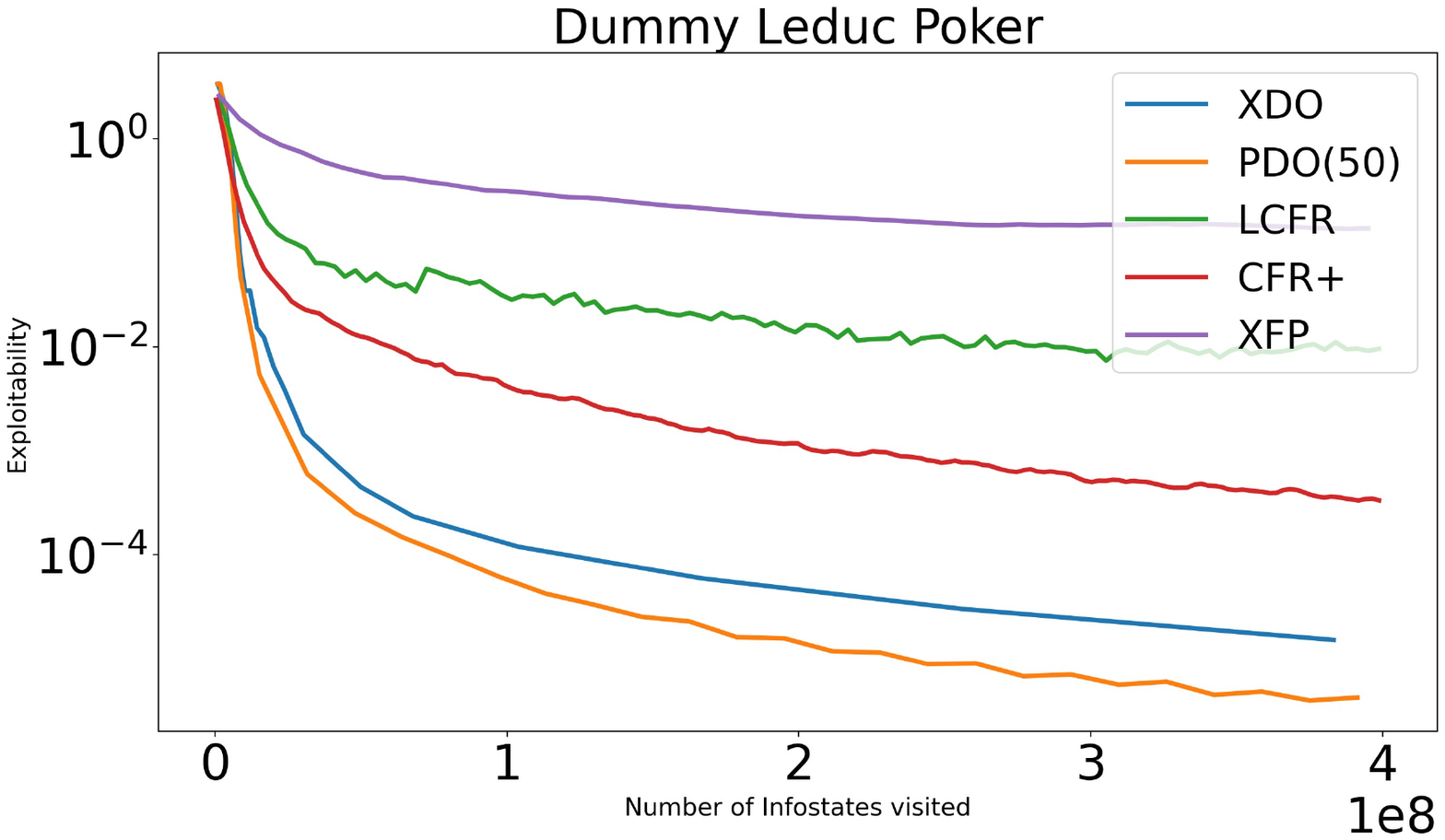}
\end{subfigure}
\begin{subfigure}
    \centering
    \includegraphics[width=.32\textwidth]{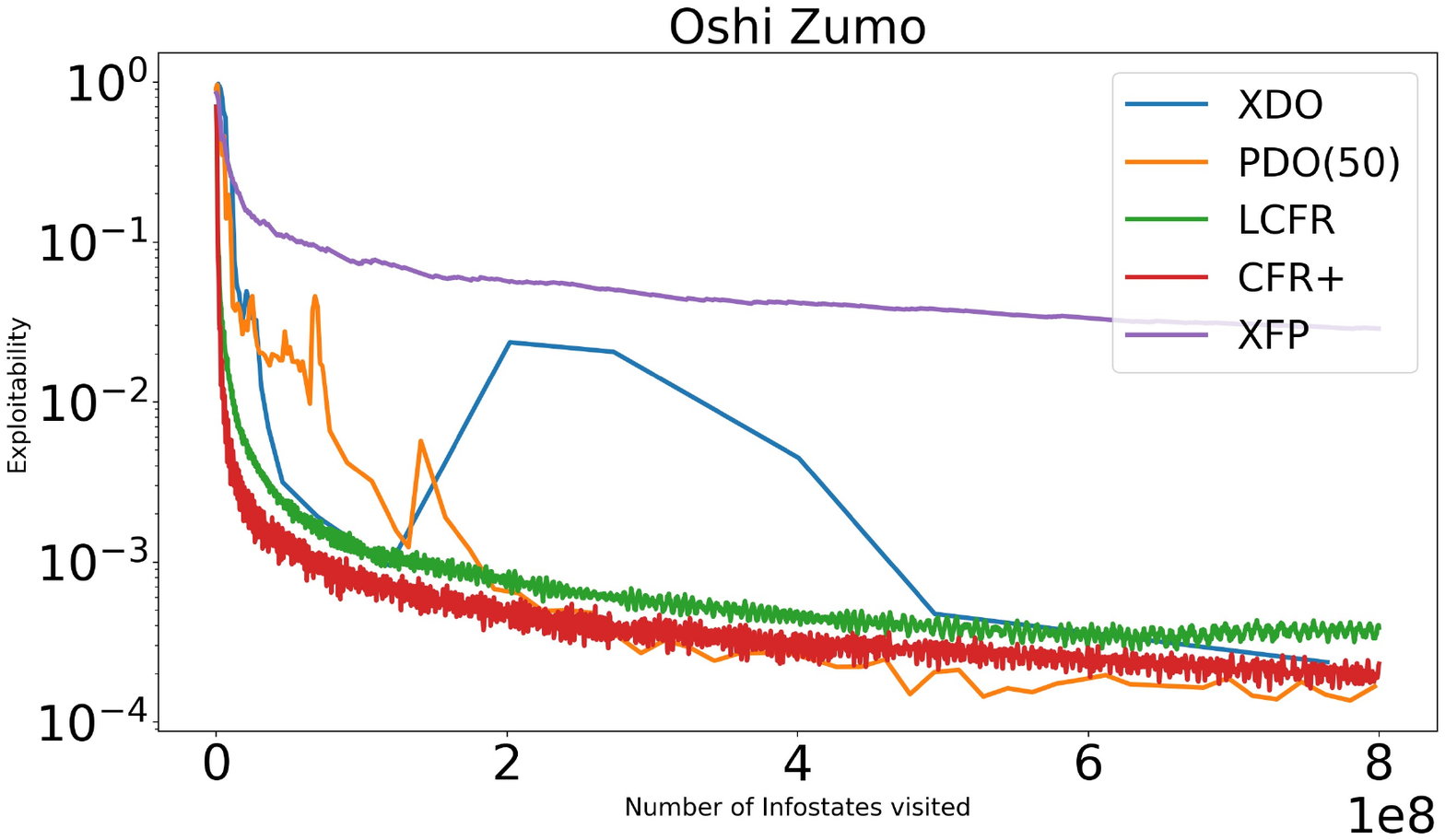}     
\end{subfigure}
\caption{Performance of PDO with periodicity $50$ and other baselines including Linear CFR (exploitability-number of infostates visited).}
\end{figure*}

\begin{figure*}[t!]
     \centering
\begin{subfigure}
    \centering
    \includegraphics[width=.32\textwidth]{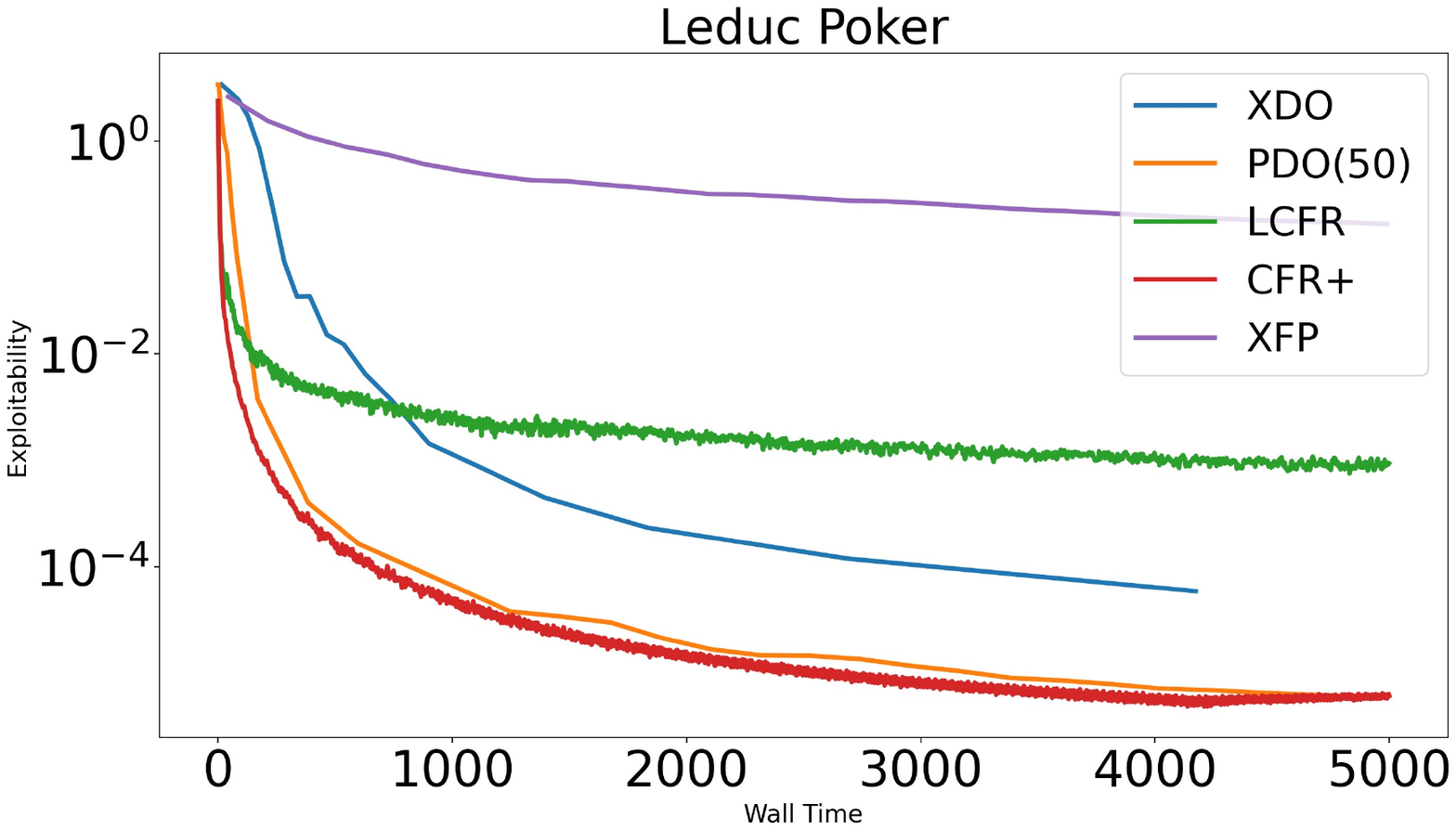}     
\end{subfigure}
\begin{subfigure}
    \centering
    \includegraphics[width=.32\textwidth]{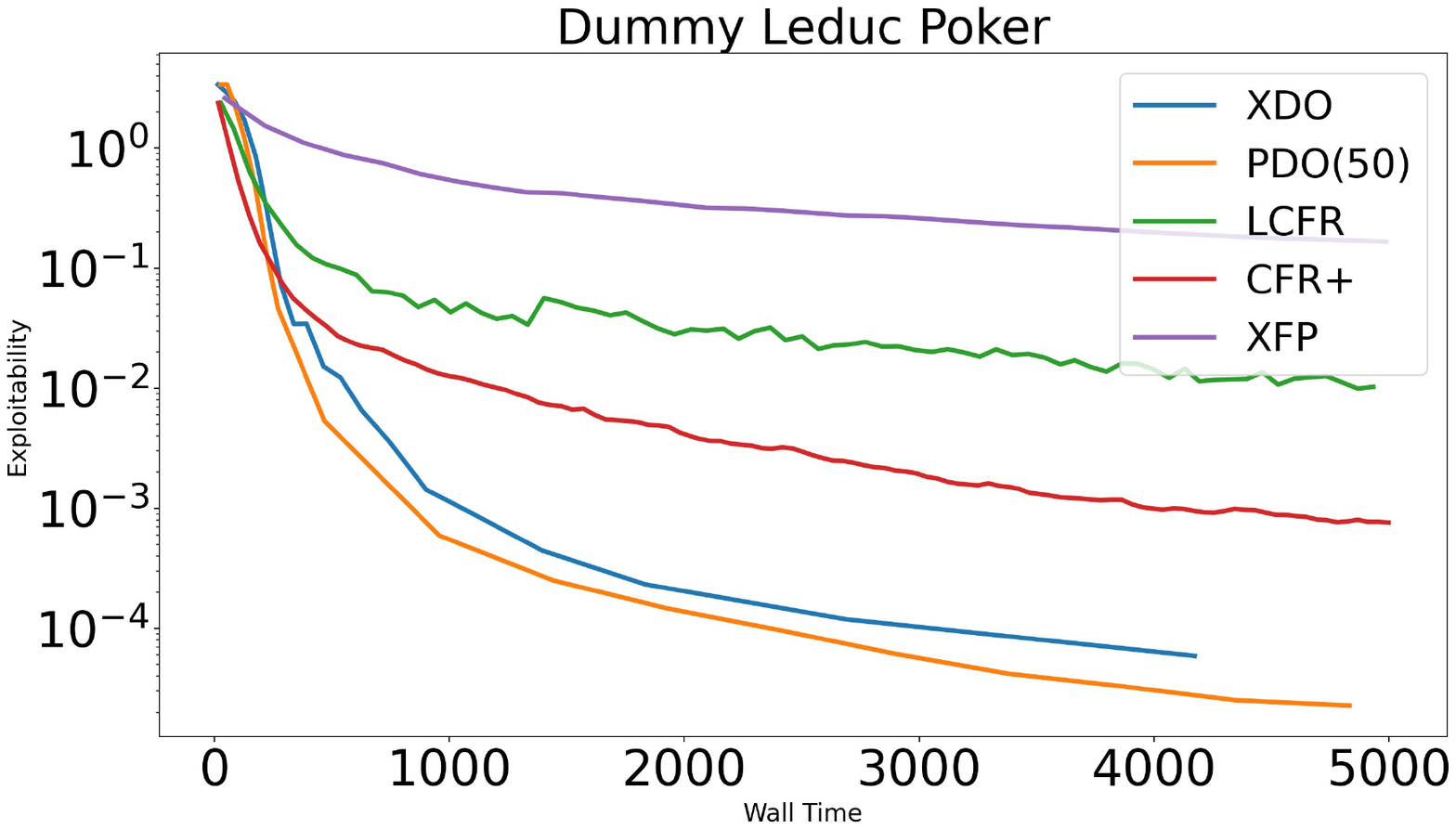}
\end{subfigure}
\begin{subfigure}
    \centering
    \includegraphics[width=.32\textwidth]{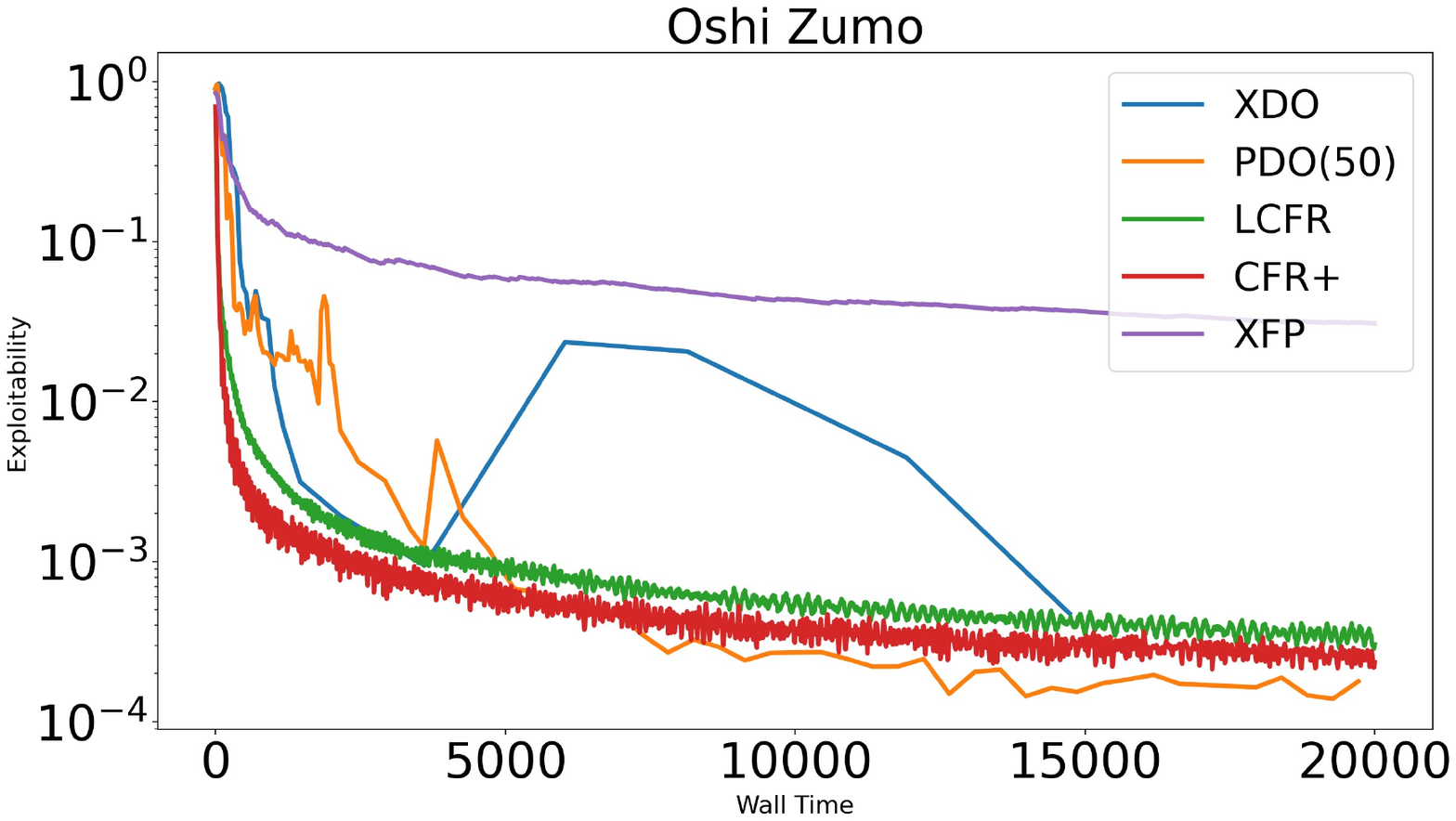}     
\end{subfigure}
\caption{Performance of PDO with periodicity $50$ and other baselines including Linear CFR (exploitability-wall time).}
\end{figure*}

\vspace{5cm}

\end{document}